\documentclass[APA,STIX1COL]{WileyNJD-v2}
\usepackage{comment}
\usepackage{subcaption} 
\usepackage{float} 
 
\articletype{Article Type}%

\newcommand{\R}{\mathbb{R}}

\newcommand{\interv}[2]{\lbrace #1, \dots, #2\rbrace}
\newcommand{\cT}{\mathcal{T}}

\newcommand{\dd}{\mathrm{d}}

\usepackage{float}

\algrenewcommand\algorithmicensure{\textbf{Output:}}
\algrenewcommand\algorithmicrequire{\textbf{Input:}} 

\title{Stochastic generator of trajectories from record data: application to the fluctuations of a glacier's frontal position from a sample of moraines}

\author[1]{Megret Maud}

\author[3]{Pereira Mike}

\author[2]{Eckert Nicolas}

\author[1]{Naveau Philippe}

\author[4]{Jomelli Vincent}

\authormark{MEGRET \textsc{et al.}}

\address[1]{\orgdiv{LSCE}, \orgname{Laboratoire des Sciences du Climat et de l'Environnement},\orgaddress{\state{Saint-Aubin}, \country{France}}}

\address[2]{\orgdiv{Université Grenoble Alpes}, \orgname{INRAE, Institut National de Recherche pour l'Agriculture, l'Alimentation et l'Environnement}, \orgaddress{\state{Saint-Martin-d'Hères}, \country{France}}}

\address[3]{
\orgname{Centre STIM, Mines Paris PSL},
\orgaddress{
\state{Fontainebleau},
\country{France}
}}
\address[4]{\orgname{CEREGE, Centre de Recherche et d'Enseignement en Géosciences de l'Environnement},\orgaddress{\state{Aix-en-Provence}, \country{France}}}

\corres{Megret Maud,
LSCE, Laboratoire des Sciences du Climat et de l'Environnement, Saint-Aubin, France, \email{maud.megret@lsce.ipsl.fr}}

\abstract[Summary]{The record values theory study elements of a time series that exceed all previous observations, which are of particular interest in fields such as sports or climate science. In this paper, we propose a statistical method based on the construction of a Brownian stochastic simulator to reconstruct entire time series solely from such record values, even in a non-stationary case.  We then  implement a procedure, which can be compared to a Neural-Based Inference (NBI) procedure, to choose the optimal generator hyper parameters. To illustrate our method and motivate its development, we apply it to a glaciological problem. Understanding the past dynamics of glacier fronts is a major challenge to mitigate related mountain hazards, assess water resources, and evaluate contributions to sea-level rise. Field-visible indicators such as moraines provide spatio-temporal evidence of these  front position evolution (refered as trajectories) and can be interpreted as the records of a non-stationary process. As a benchmark case, the two hyper parameters of our NBI approach are tuned from the well documented French alpine \textit{Glacier des Bossons}. Our purely data-based approach offers new perspectives for challenging and further developing physical models of glacier dynamics and inferring the response of glaciers to climate change on centennial to millenial time scales. Beyond the glacier case, it has potential for the various problems for which record series is the sole available data.}

\keywords{Record theory, Stochastic simulations, Brownian Motion,  Neural Network, Glacier, Moraine, Climate Change}

\jnlcitation{
\cname{
\author{M. Megret},
\author{M. Pereira},
\author{N. Eckert},
\author{P. Naveau},
and \author{V. Jomelli}}
(\cyear{2026}),
\ctitle{Stochastic generator of trajectories from records data}}

\begin{document}
\maketitle

\section{Introduction}\label{intro}
In time series analysis, the theory of records focuses on modeling the probability distributions of observations that exceed   all previous values. 
It has been applied to  multiple fields, including  environmental studies \cite[see, e.g.,][]{majumdarExactlySolvableRecord2019,cebrian2022record,Fischer25},   and, in parallel,  mathematical developments that characterize record properties  have been steadily growing  over the last decades \cite[see, e.g.,][]{ahsanullahRecordsProbabilityTheory2015}. 
For instance, \cite{majumdarUniversalRecordStatistics2008a} investigated the frequency of records in the context of random walks \cite[see, also][]{wergenRecordsStochasticProcesses2013}, deriving the expectation and variance of total number  of records over a given number of  steps. This type of work was pursued by   \cite{godrecheRecordStatisticsStrongly2017} who calculated the shortest expected time between two consecutive records, followed by additional advances focusing on random walks \cite[see, e.g.,][]{mounaixStatisticsNumberRecords2020}. 
Moving away from random walks, \cite{ballerini_records_1987} proposed a  theoretical treatment of linear trends in records \cite[see also][]{Feuerverger96}, while \cite{Hoayek17} studied the non-stationary Yang-Nevzorov model and 
\citet{Gonzalez24} dealt with non-stationarity under the umbrella of extreme value theory \cite[see, e.g.,][]{beirlant_statistics_2004}. 

Despite these probability and statistical advances, major hurdles remain concerning the modeling of records in some applied contexts. In this article, our motivation stems from the following  practical problem. Given only records times and positions, our question is to determine how to simulate stochastic trajectories that are in compliance with such records. 
Here, being in compliance means that, not only a valid trajectory has to go through the given records (in time and in intensities), but this trajectory cannot produce new records that could erase the given ones. 
The case-study situation will be detailed in Section \ref{application} where we will see that glacier moraines can be interpreted  as records produced by glacier advances (see Figure \ref{fig:plot_mor}). Our  main applied question is to determine if the timing and positions of such records will allow us to  reconstruct (simulate) realistic past glacial trajectories on a long time scale over which the entire trajectory was not monitored systematically, a key knowledge in glaciology with various implications for, e.g. climate change assessment, water resources, ecology and mountain hazards \cite[see, e.g.][]{zemp2019global,cauvy2019global,jomelli2011irregular}. From a computational side, the  task of generating  stochastic trajectories under constraints   belongs to the so-called domain of conditional simulations \cite[see, e.g.,][]{Lantuejoul2002}. 
To our knowledge, current  conditional simulation techniques \cite[see, e.g.,][]{Wackernagel2003} are not tailored to respect the natural order present in records. However, there are methods that come close to this, for example, some research examines Gaussian processes with inequality constraints (\cite[see, e.g.][]{da2012gaussian,bachoc2022sequential}). These are sequential approaches, often based on rejection criteria, are, in our case, are rendered too cumbersome by the nested nature of the inequalities.  There are also machine learning-based generative time series models (see, e.g.,\cite{coletta2023constrained}) that work by learning inequality constraints from the training data; however, in our case, we assume that we do not necessarily have a sufficiently large dataset.

In this work, we propose an interpretable model based on the simple concept of a random walk, which can be fitted to a limited amount of data and also allows for the rapid generation of trajectories, given records and two hyperparameters. These hyper parameters may be chosen based on expert information, or using an inference scheme based on partial trajectory observations, which we illustrate  on the \textit{Glacier des Bossons} dataset (\citep{nussbaumer2012little}).\\
To summarize our plan,  Section \ref{model} focuses on describing our assumptions, our model  and a stochastic simulation scheme from a set of given records positions. 
A numerical application to the analysis of past glacial trajectories with respect to moraines and the choice of the hyper parameters are detailed in Section \ref{application}.

\section{Simulating  trajectories from given records}\label{model}
 
In a seminal paper, \cite{renyiTheorieElementsSaillants} formally defined the notion of records of a time series. Given a time series sequence of real-valued random variables, say $\{X_1,X_2, \dots\}$,  the value   $X_t$ is said to be a record   if 
\begin{equation}\label{def: records}
    X_t > X_j, \mbox{ for all positive integers  } j \in S_t,
\end{equation}
where  $S_t$ represents a given  reference set of indices. 
For example, the set $S_t= \{0, 1, \dots, t-1\}$ (with the usual convention $X_0=-\infty$) corresponds to the classical notion of records for which all past values have to be  smaller than the current one. 
In this specific context, \citeauthor{renyiTheorieElementsSaillants} showed that, whenever all variables $\{X_1,X_2, \dots\}$ are  assumed to be  independent and identically distributed (i.i.d.),  
$ \mathbb{P}(X_t>\max(X_0,...,X_{t-1}))= \frac{1}{t}$ for any positive integer $t$.
But the choice of $S_t$ can influence this property. Of particular interest in this work is the case where the set $S_t$ is taken as $S_t= \{0, 1, \dots, t-1, t+1\}$. This choice encodes a different viewpoint on the notion of record: a record $X_t$ is considered not only as the maximum value attained by the time series up to time $t$ ($X_t>\max\{X_0,...,X_{t-1} \}$), but also immediately following $t$ ($X_t> X_{t+1}$). Hence, such a record can be interpreted as marking the endpoint of an increasing segment of the time series that reaches a value higher than any previously observed value. This notion of records is directly related to the concept of moraines, the motivating example of this study  (see Section \ref{application}). Indeed, moraines are geological markers of the end of a glacier’s advance and thus correspond to records in the time series of glacier front positions. Figure \ref{fig:intro} compares the two records types, obtained with the classical set $S_t= \{0, 1, \dots, t-1\}$ (orange dots), and  with the set $S_t= \{0, 1, \dots, t-1, t+1\}$ (blue circles). As one can notice, in the classical definition, all times leading to record time, and during which the time series is non-decreasing, are also record times. However, in the latter case, this is now impossible as having records at two consecutive times is by construction, impossible, and records must now be local maxima of the time series. Moreover, as the set $ \{0, 1, \dots, t-1, t+1\}$ contains $ \{0, 1, \dots, t-1\}$, our newly defined records are necessarily records in the classical sense. However, under the i.i.d.  assumption, and considering $S_t=\{0, 1, \dots, t-1, t+1\}$, we now have
\begin{equation*}
   \mathbb{P}(X_t>\max\{X_0,...,X_{t-1},X_{t+1}\})= \frac{1}{t+1} \neq \frac{1}{t}
\end{equation*}  

\begin{center}
    \includegraphics[width=0.7\linewidth]{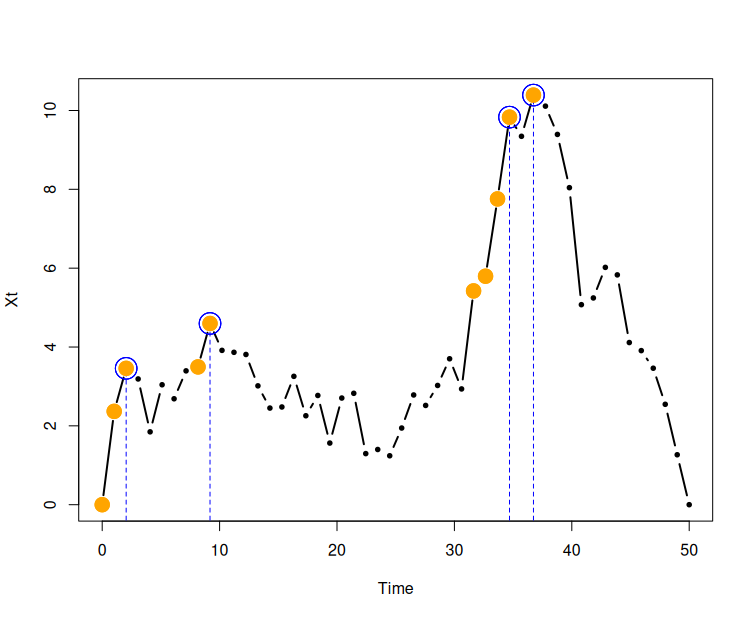}
    \captionof{figure}{Influence of the    reference set $S_t$ choice in record definition \eqref{def: records} from a simulated time series.  The orange dots corresponds to records defined with  
    $S_t=\{0, 1, \dots, t-1\}$, while the blue circles occur when $S_t=\{0, 1, \dots, t-1, t+1\}$. }
    \label{fig:intro}
\end{center}

Most  of the  literature on records is based on the set $S_t=\{0, 1, \dots, t-1\}$, but otherwise, it is  very sparse. This limits the applicability  of mathematical tools, in particular  for our record application (inferring glacial retreats from moraines, see Section \ref{application} for details).  Hence, one methodological contribution of the present work is to show that modeling and  simulating records  with $S_t = \{0, 1, \dots, t-1,t+1\}$ can be also implemented. \\

Let us introduce some notations and definitions. We call \textit{record set} a set of $N\ge 2$ pairs $\mathcal{M}=\{(R_n,L_n)\}_{n=1,\dots,N} \subset \R\times \R_+$ such that, for any $1\le n<N$, $L_{n+1}>L_n+1$ and $R_{n+1}>R_n$. Starting now from a time series $\{X_t\}_{t\in\interv{1}{\mathcal{T}}}$, $\mathcal{T}>1$, we write $\mathcal{R}(\{X_t\}_{t\in \interv{1}{\mathcal{T}}})$, the set of its records, marked as pairs (Value of the record, Occurence time of the record). In particular,  $\text{Card}~\mathcal{R}(\{X_t\}_{t\in \interv{1}{\mathcal{T}}})\ge 2$ and $ \mathcal{R}(\{X_t\}_{t\in \interv{1}{\mathcal{T}}})$ is by construction a record set as defined above. Our aim is to do the converse: given a record set $\mathcal{M}=\{(R_n,L_n)\}_{n=1,\dots,N} $, we aim at simulating a stochastic trajectory $\{X_t\}_{t\in\interv{1}{\mathcal{T}}}$ such that $\mathcal{M} \subset \mathcal{R}(\{X_t\}_{t\in \interv{1}{\mathcal{T}}})$. That means that all records defined in $\mathcal{M}$ are records of $\{X_t\}_{t\in\interv{1}{\mathcal{T}}}$ but this time series may admit additional records. Moreover, we will make the following assumption: $L_1=1$ and $L_N=\mathcal{T}$, meaning that the time series we model starts at time $t=L_1=1$ with the first record in $\mathcal{M}$ (i.e. $X_1=R_1$), and ends at time $t=L_N=\mathcal{T}$ with the last record in $\mathcal{M}$ (i.e. $X_\mathcal{T}=R_N$). We postpone to the end of Section~\ref{step2} the discussion about the case where  $L_1>1$ and/ot $L_N<\mathcal{T}$.

\subsection{Case 1: $\mathcal{M}$ contains two records}

Given $\mathcal{T}>3$ and a record set $\mathcal{M}=\lbrace (R_1,1), (R_2,\mathcal{T})\rbrace$ we aim at modeling times series $(X_t)_{t\in\interv{1}{\mathcal{T}}}$ with endpoints $(R_1,1)$ and $(R_{2},\mathcal{T})$, and such that these pairs are records of the time series. To do so, we rely on a well-known, fast  and   simple building block  to simulate  a random trajectory from a given starting point to a given end point: the Brownian motion (or Wiener process). For $T>0$, it is a continuous-time stochastic process, say $(B_t)_{t\in [0,T]}$, with the following  properties. It is  (almost-surely) continuous in time,   $B_0=0$, and it has 
zero-mean Gaussian independent increments, i.e.,  $B_t-B_s\sim\mathcal{N}(0,\sigma^2(t-s))$ for any  $0\le s<t$, and $B_t - B_s$ is independent of $B_{t'} -B_{s'}$ whenever $0\le s<t\le s'<t'$ \cite[see, e.g.,][Definition 1.1]{morters2010brownian}. The standard deviation $\sigma$ is set to one in the classical definition, but a  scaling parameter brings an added flexibility needed in our application. When discretized in time (with constant time steps $1$), the Brownian motion $(B_t)_{t\in \lbrace 0,1,\dots}$ can be seen as a random walk with independent Gaussian increments $\mathcal{N}(0,\sigma^2)$. A stochastic process related to the Brownian motion, and closer to the notion of records, is the Brownian bridge, which corresponds to a Brownian motion conditioned to be to take the same value as its initial value, i.e. 0, at some fixed future time $T>0$. Indeed, starting from a Brownian motion $(B_t)_{t\in[0,T]}$, a Brownian bridge $\tilde{B}_t$ on $[0,T]$ can be constructed through the transformation
\begin{equation}\label{eq:bridge_std1}
    \tilde{B}_t = B_t - \frac{t}{T}B_T, \quad t\in [0,T].
\end{equation}
Further, a Brownian bridge $(\tilde{B}_t^{x,y})_{t\in[0,T]}$ with arbitrary endpoints $\tilde{B}_0^{x,y}=x$ and $\tilde{B}_T^{x,y}=y$ can be constructed via a simple affine transform of \eqref{eq:bridge_std1}, thus giving
\begin{equation}\label{eq:bridge_std2}
    \tilde{B}_t^{x,y} = x+\frac{t}{T}(y-x)+\tilde{B}_t , \quad t\in [0,T].
\end{equation}
Such a construction then naturally yields a (continuous-time) stochastic process in $[0,T]$ that is conditioned to start and end  at predefined values \citep[Chapter I.3]{revuz2013continuous}. This is the starting point of our stochastic simulator. Indeed, let us circle back to our simulation problem. First, we need to generate a time series $(X_t)_{t\in\interv{1}{\mathcal{T}}}$ that which starts $X_1=R_1$, ends with $X_T=R_T$ and admits the pairs $(R_1,1)$ and $(R_2,\mathcal{T})$ as records. To fulfill the first two requirements, we propose to simulate a discrete Brownian bridge $(\tilde{B}_t^{R_1,R_2})_{t\in[0,T]}$, conditioned to start at $R_1$ and end at $R_2$, on a time range $T=\mathcal{T}-1$. This can be done by simply generating a discrete Brownian motion (with time step $1$), and applying the transformations \eqref{eq:bridge_std1} and \eqref{eq:bridge_std2}.  A trajectory $(Y_t)_{t\in \{1,...,\mathcal{T}\}}$ that starts at the record  $(R_1,1)$ and ends at the second one $(R_{2},\mathcal{T})$ is then obtained by simply setting
\begin{equation}\label{eq:bridge_mod}
    Y_t = \tilde{B}_{t-1}^{R_1,R_{2}}, \quad t\in \{1,...,\mathcal{T}\} .
\end{equation} 

There is no guarantee however that the resulting trajectory actually admits $(R_1,1)$ and $(R_2,\mathcal{T})$ as records, which would mean in particular that $Y_{1}$ is a local maximum and $Y_{\mathcal{T}}$ a global maximum of $(Y_t)_{t\in \interv{1}{\mathcal{T}}}$. To enforce this,  we apply a Vervaat-type transform (\cite{lupuVervaatTransformBrownian2015})  to the trajectory . Such transforms are classically used to construct, from a Brownian bridge, a so-called Brownian excursion, i.e. a Brownian bridge conditionned to take positive values. In our setting, we consider the  transform $V$ that maps a time series $\lbrace F_t\rbrace_{t\in\interv{1}{\mathcal{T}}}$ to a time series $\lbrace V(F)_t\rbrace_{t\in\interv{1}{\mathcal{T}}}$ given by
\begin{align}\label{eq:vervaat}
t &\mapsto V(F)_t=\left\{ \begin{array}{ll}F_{1}-F_{T_{max}}+F_{T_{max}-1+t}, &\mbox{if }1 \leq t \leq 1+\mathcal{T}-T_{max},\\
    F_{\mathcal{T}}-F_{T_{max}}+F_{T_{max}-\mathcal{T}+t}, &\mbox{if }  1+\mathcal{T}-T_{max}<t\leq \mathcal{T},
    \end{array}
    \right.
\end{align}
where 
\begin{equation}\label{eq:Tmax}
    T_{max}= \min\left\lbrace\mathop{\text{argmax}}_{t\in \interv{1}{\mathcal{T}}}{F_t}\right\rbrace,
\end{equation}
i.e. $T_{max}$ corresponds to the smallest time in $\interv{1}{\mathcal{T}}$ realizing the maximum of the time series $\lbrace F_t\rbrace_{t\in\interv{1}{\mathcal{T}}}$ . 
In particular, if $T_{max}=\mathcal{T}$, then $V(F)_t=F_t$ for any $t\in\interv{1}{\mathcal{T}}$ and the endpoint $F_{\mathcal{T}}$ is necessarily a record of $\lbrace F_t\rbrace_{t\in\interv{1}{\mathcal{T}}}$. Besides,
one can show (cf. Appendix~\ref{appen:vervaat} for details) that, assuming that $F_1<F_{\mathcal{T}}$ and that $T_{max}<\mathcal{T}$, the process $(V(F)_t)_{t\in \interv{1}{\mathcal{T}}}$ has the same endpoints as $F$ (i.e  $V(F)_{1}=F_{1}$ and $V(F)_{\mathcal{T}}=F_{\mathcal{T}}$). Besides, these endpoints are local maxima of the time series, and satisfy in particular $V(F)_t\leq V(F)_{\mathcal{T}}$ for $t\in\interv{1}{\mathcal{T}}$. It is then possible to see that, when applying the transform~\eqref{eq:vervaat} to the modified bridge $(Y_t)_{t\in\interv{1}{\mathcal{T}}}$ obtained in \eqref{eq:bridge_mod}, we  obtain a trajectory
\begin{equation}\label{eq:traj_rec}
    Z_t = V(Y)_t, \quad t\in\interv{1}{\mathcal{T}}
\end{equation}
which by construction admits $(R_1,1)$ and $(R_2,\mathcal{T})$ as records (and as endpoints) whenever the corresponding time $T_{max}$ satisfies $T_{max}<\mathcal{T}$. A simple rejection approach based on the value of $T_{max}$ can ensure that this last condition is satisfied, thus resulting in the following simulation approach yielding a time series conditioned by a pair of records which are also its endpoints:
\begin{enumerate}
    \item Generate a bridge $(Y_t)_{t\in\interv{1}{\mathcal{T}}}$ as defined in~\eqref{eq:bridge_mod} between the two records
    \item Compute the associated $T_{max}$ in~\eqref{eq:Tmax}
    \item \textbf{If} $T_{max}=\mathcal{T}$ And $Y_2<Y_1$ (i.e. $Y_1$ is a record): Return the time series $(Y_t)_{t\in\interv{1}{\mathcal{T}}}$\\
    \textbf{Else If} $T_{max}<\mathcal{T}$ : Return the time series $(Z_t)_{t\in\interv{1}{\mathcal{T}}}$  defined as the Vervaat transform~\eqref{eq:vervaat} of $(Y_t)_{t\in\interv{1}{\mathcal{T}}}$\\
    \textbf{Else}: Go to 1.
\end{enumerate}

\begin{remark}
    The acceptance rate of the rejection step above can be approximately evaluated using the properties of (time-continuous) Brownian bridges. This is discussed in Appendix~\ref{appen:rejection}, where we show in particular that by either working with a large enough time series size $\cT$, and/or with a large enough variance $\sigma^2$ for the Brownian bridge, the acceptance rate can be controlled.
\end{remark}

Finally, to allow different degrees of smoothness on trajectories, we apply  a  convolution using a Gaussian kernel of  parameter $s$ to $\lbrace Z_t\rbrace_{t\in\interv{1}{\mathcal{T}}}$.  
 For $\rho>0$ and $k\in\R$, take 
 \begin{equation*}
     w^{(s)}(k)=\exp\left(-\frac{k^2}{2s^2}\right)
 \end{equation*}
We then define the associated smoothed time series $\lbrace\widetilde{Z}_t\rbrace_{t\in\interv{1} {\mathcal{T}}}$ by
 \begin{equation*}
     \tilde{Z}_t = \frac{1}{c_t}\sum_{j=1}^{\mathcal{T}}w^{(s)}(j-t)Z_{j}, \quad\text{where}\quad c_t = \sum_{j=1}^{\mathcal{T}}w^{(s)}(j-t).
 \end{equation*}
The time series $\lbrace\widetilde{Z}_t\rbrace_{t\in\interv{1} {\mathcal{T}}}$ corresponds to a Gaussian smoothing of the time series $\lbrace {Z}_t\rbrace_{t\in\interv{1} {\mathcal{T}}}$. But as such, the resulting time series $\lbrace\widetilde{Z}_t\rbrace_{t\in\interv{1} {\mathcal{T}}}$ does not have the same endpoints as the original time series $\lbrace{Z}_t\rbrace_{t\in\interv{1} {\mathcal{T}}}$. To smooth the behavior of the smoothed time series near the extremities and ensure that its endpoints are same as the original time series, we propose a final adjustment which consists in assuming a Gaussian-like decay of the time series near its endpoints. We end up with a final smoothing $\lbrace X_t\rbrace_{t\in\interv{1}{\mathcal{T}}}$ of the time series $\lbrace Z_t\rbrace_{t\in\interv{1}{\mathcal{T}}}$  defined as
\begin{equation}\label{eq:w_s}
   X_t
   =w^{(s)}(t-1)\frac{1-w^{(s)}(\mathcal{T}-t)}{1-w^{(s)}(\mathcal{T}-1)}Z_1 + \left(1-w^{(s)}(\mathcal{T}-t)\right)\left(1-w^{(s)}(t-1)\right)\widehat{Z}_t+w^{(s)}(\mathcal{T}-t)\frac{1-w^{(s)}(t-1)}{1-w^{(s)}(\mathcal{T}-1)}Z_{\mathcal{T}} 
\end{equation}
where we take for $t\in\interv{1}{\mathcal{T}}$,
\begin{equation}\label{eq:def_hatz}
    \widehat{Z}_t=\left\lbrace\begin{aligned}
        &\min\left\lbrace \widetilde{Z}_2 \,,\frac{\alpha_{\mathcal{T}}}{1-w^{(s)}(\mathcal{T}-1)}Z_1\right\rbrace &&\text{if } t=2 \\
        &\widetilde{Z}_t &&\text{otherwise}
    \end{aligned}\right., \quad \text{where } \alpha_{\mathcal{T}}=(1-w^{(s)}(\mathcal{T}-3))+w^{(s)}(\mathcal{T}-3)\frac{w^{(s)}(\mathcal{T}-1)}{w^{(s)}(\mathcal{T}-2)}.
\end{equation}
This means in particular that we have $X_t \approx w^{(s)}(t-1)Z_1 $ when $\vert t-1\vert/s\rightarrow 0$, $X_t \approx w^{(s)}(\mathcal{T}-t)Z_\mathcal{T} $ when $\vert \mathcal{T}-1\vert/s\rightarrow 0$, and $X_t \approx \tilde{Z}_t$ ($=$ Gaussian smoothing of  $\lbrace Z_t\rbrace_{t\in\interv{1}{\mathcal{T}}}$ with bandwidth $s$)  when $\vert \mathcal{T}-t\vert \gg s$ and $\vert t-1\vert \gg s$. It is then possible to check that, under some mild restrictions on the record values which are detailed and explained in Appendix~\ref{appen:smooth}, the endpoints of the time series $\lbrace X_t\rbrace_{t\in\interv{1}{\mathcal{T}}}$ are indeed $(R_1,1)$ and $(R_2,\mathcal{T})$, and are records of the time series. \\

To conclude this subsection, we summarize in Algorithm~\ref{alg1} the approach used to simulate a random trajectory $(X_t)_{t\in\interv{1}{\mathcal{T}}}$ conditioned by a record set $\mathcal{M}=\lbrace (R_1,1),(R_2,\mathcal{T})\rbrace$ defining its endpoints.

\begin{algorithm}
\caption{Simulation algorithm conditioned to $2$ records}\label{alg1}
\begin{algorithmic}
\Require  Time horizon $\mathcal{T}>3$, Record set $\mathcal{M}=\lbrace (R_1,1),(R_2,\mathcal{T})\rbrace$, Standard-deviation of Gaussian increments $\sigma>0$, Smoothing parameter $s\ge 0$ all chosen so that~\eqref{eq:cond_rec} is satisfied
\Ensure Simulated trajectory $(X_t)_{t\in\interv{1}{\mathcal{T}}}$ such that $\mathcal{M}\subset \mathcal{R}((X_t)_{t\in\interv{1}{\mathcal{T}}})$
    \State $\mathtt{restart}\leftarrow\mathtt{TRUE}$
    \While{$\mathtt{restart}$}
  \State Generate a discretised Brownian motion $(B_t)_{t\in\interv{0}{\mathcal{T}-1}}$  with variance of increments $\sigma^2$
  \State Transform  $( B_t)_{t\in\interv{0}{\mathcal{T}-1}}$ into a Brownian Bridge $(\tilde B_t)_{t\in\interv{0}{\mathcal{T}-1}}$using~\eqref{eq:bridge_std1}
  \State Transform $(\tilde B_t)_{t\in\interv{0}{\mathcal{T}-1}}$ into a modified bridge $\tilde{B}_t^{R_1,R_{2}}$ starting at $R_1$ and ending at $R_2$ using~\eqref{eq:bridge_std2}
  \State Compute the time series $(Y_t)_{t\in\interv{1}{\mathcal{T}}}$ defined as $Y_t\leftarrow \tilde{B}_{t-1}^{R_1,R_{2}}$
  \State Recompute $T_{max}$ in~\eqref{eq:Tmax} on the time series 
  $(Y_t)_{t\in\interv{1}{\mathcal{T}}}$ 
   \If{$T_{max}=\cT$ \textbf{And} $Y_2<Y_1$} 
   \State Set $Z_t=Y_t$ for $t\in\interv{1}{\cT}$
 \State $\mathtt{restart}\leftarrow 
 \mathtt{FALSE}$
 \ElsIf{$T_{max}<\mathcal{T}$}
  \State Compute $\lbrace Z_t\rbrace_{t\in\interv{1}{\mathcal{T}}}$  defined as the Vervaat transform~\eqref{eq:vervaat} of $(Y_t)_{t\in\interv{1}{\mathcal{T}}}$ 
  \State $\mathtt{restart}\leftarrow 
 \mathtt{FALSE}$
 \EndIf
  \EndWhile
 \State Compute $(X_t)_{t\in\interv{1}{\mathcal{T}}}$ defined as the Smoothing~\eqref{eq:w_s} with parameter $s$ of $(Z_t)_{t\in\interv{1}{\mathcal{T}}}$ 
  \State\Return $(X_t)_{t\in\interv{1}{\mathcal{T}}}$ 
\end{algorithmic}
\end{algorithm}

\subsection{Case 2: $\mathcal{M}$ contains more than two records}\label{step2}

Given now a record set $\mathcal{M}=\lbrace (R_n,L_n)\rbrace_{n\in\interv{1}{N}}$ (with $L_1=1$ and $L_n=\mathcal{T}$), we aim at modeling times series $(X_t)_{t\in\interv{1}{\mathcal{T}}}$ with endpoints $(R_1,1)$ and $(R_{N},\mathcal{T})$, and such that the pairs in $\mathcal{M}$ are records of the time series. At this stage, we know how to simulate trajectories between each pair of consecutive records, $(R_n,L_n),(R_{n+1},L_{n+1})$ for $n \in \{1,...,N-1\}$, using Algorithm~\ref{alg1} with $\mathcal{T}=L_{n+1}-L_n+1$ and $\mathcal{M}=\lbrace (R_n,1),(R_{n+1},L_{n+1}-L_n+1)\rbrace$. Hence, to obtain a complete trajectory  $\{X_t\}_{t \in\interv{1}{\mathcal{T}}}$ we propose to simply assemble the smoothed transformed trajectories~\eqref{eq:w_s} built between each pair of records $(R_n,L_n),(R_{n+1},L_{n+1})$, thus giving
\begin{equation}
    X_t = S_{t-L_n+1}^n, \quad \text{for } t\in \interv{L_n}{L_{n+1}},
\end{equation}
where $1\le n<N$ and by construction (cf. previous subsection) $S_{L_{n+1}-L_n+1}^n=S_1^n$. The overall approach to generate time series conditioned to a set of records $\mathcal{M}$ is summarized in Algorithm \ref{alg2}. As defined, the time series $(X_t)_{t\in\interv{1}{\mathcal{T}}}$ satisfies $\mathcal{M}\subset\mathcal{R}(\{X_t\}_{t\in\interv{1}{\mathcal{T}}})$, as shown in details in the Appendix~\ref{appen:res_alg2}, and thus answers the simulation question of this subsection. This approach is once again justified by analogy the properties of Brownian motions, and in particular their Strong Markov property. Indeed, considering our record occurrence times in $\mathcal{M}$ as stopping times, simulations of the process after said stopping time (but conditioned to start at the same place) are independent of the behavior prior to the stopping time \cite[Theorem 2.16]{morters2010brownian}.\\

Note that two parameters,   $\sigma$ and $s$, are needed to implement the proposed algorithm.  
The first one scales the Brownian process $B_t$, while the second one captures the smoothness of the kernel as explained in \eqref{eq:w_s}. We can now move on to explain the role of $\sigma$ and $s$.

\begin{algorithm}
\caption{Simulation algorithm conditioned to $N\ge 2$ records}\label{alg2}
\begin{algorithmic}
\Require  Time horizon $\mathcal{T}>3$, Record set $\mathcal{M}=\lbrace (R_n,L_n)\rbrace_{n=1,\dots,N}$ (with $L_1=1$, $L_N=\mathcal{T}$), Standard-deviation of Gaussian increments $\sigma>0$, Smoothing parameter $s\ge 0$
\Ensure Simulated trajectory $(X_t)_{t\in\interv{1}{\mathcal{T}}}$ such that $\mathcal{M}\subset \mathcal{R}((X_t)_{t\in\interv{1}{\mathcal{T}}})$
\State $N \leftarrow length(\mathcal{M})$ 
  \For {each $n$ in $\{1,...,N-1\}$}
  \State $l_n \leftarrow L_{n+1}-L_n+1$ 
  \State $\lbrace S^n_t\rbrace_{t\in\interv{1}{l_n}} \leftarrow$ Output of Algorithm~\ref{alg1} with: $\mathcal{T}=l_n,\mathcal{M}=\lbrace (R_n,1),(R_{n+1},l_n))\rbrace,\sigma=\sigma,s=s$
  \EndFor
  \State\Return $(X_t)_{t\in\interv{1}{\mathcal{T}}}$ defined as $X_t = S_{t-L_n+1}^n, \quad \text{for } t\in [L_n,L_{n+1}]$
  
\end{algorithmic}
\end{algorithm}

We end this subsection by noting that the simulation approach can be extended to cases where $L_1<1$ and/or $L_N<\mathcal{T}$. In the latter case, new records may be created at times greater than $\mathcal{T}$. Furthermore, the evolution after time $L_N$ does not alter the presence or absence of records prior to time $L_N$. It is therefore sufficient to generate a discretised Brownian motion from $L_N$ up to $\mathcal{T}$, conditioned on having a decreasing initial time step to preserve the fact that the point $(R_N,L_N)$ is a record (which amounts to considering an initial time step drawn from a truncated Gaussian). In the first case, the value of $R_1$ represents a threshold, and all values for times less than $L_1$ must be under it. By analogy with the construction we have proposed, we can generate such a series using the concept of a meander, a Brownian process conditioned to be less than a value, here $R_1$ (see, e.g., \cite{iglehart1974functional}).    

\subsection{Illustrations of the effects of the hyper parameters on simulated trajectories}
In this subsection, we want to give an overview of the simulations that can be made with the first step of Algorithm \ref{alg1} and of the effects that the hyper parameters $s$ and $\sigma$ have on them.
We have two of the records accepted by the $X_t$ series shown in Figure \ref{fig:intro}. They have the following coordinates, $(R_2=4.6,L_2=9),(R_3=9.8,L_3=34)$. We denote by $l$ the number of timesteps between those two records. Here, $l=L_3-L_2+1=26$. We want to simulate and analyze Z trajectories that reproduce the portion of the X trajectory between the two first records i.e. between the two blue dotted lines. 
\begin{figure}[ht]
    \centering
    \begin{subfigure}[t]{0.49\textwidth}
        \centering
        \includegraphics[width=\linewidth]{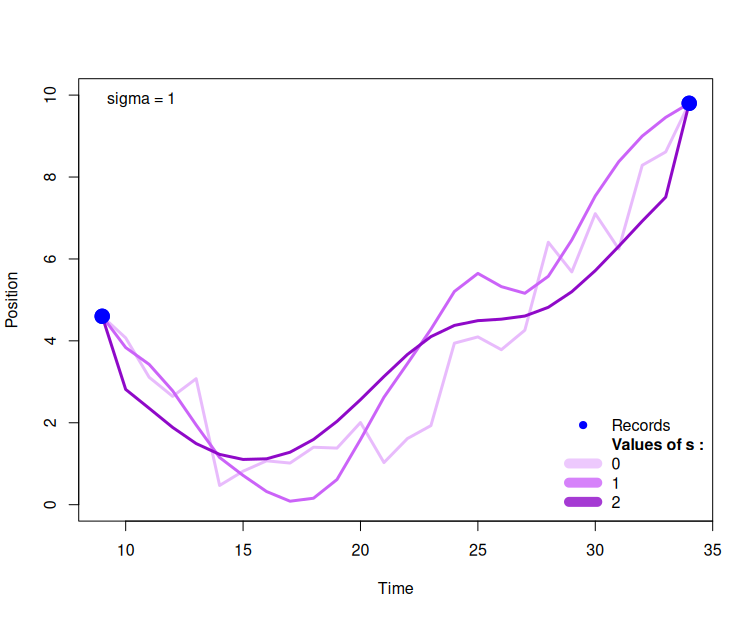}
        \caption{Influence of $s$ with $\sigma = 1$}
        \label{fig:ex_s}
    \end{subfigure}
    \hfill
    \begin{subfigure}[t]{0.49\textwidth}
        \centering
        \includegraphics[width=\linewidth]{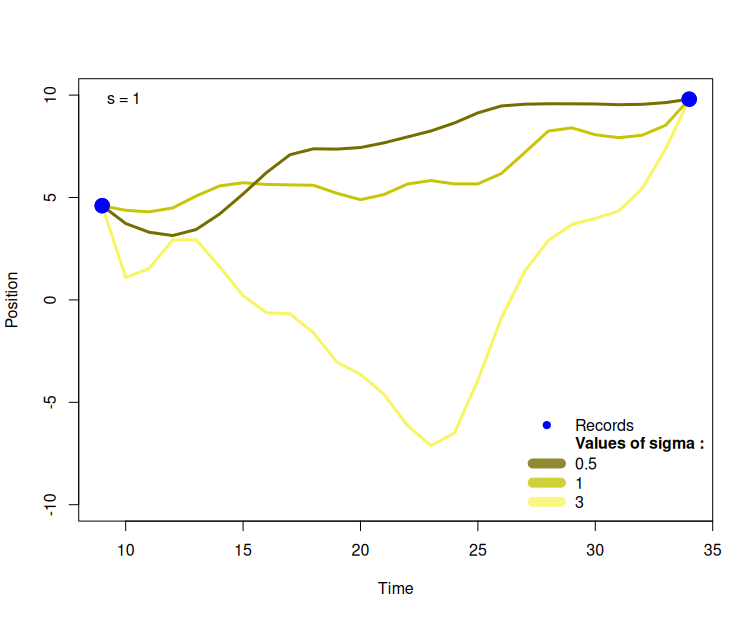}
        \caption{  Influence of $\sigma$ with $s = 1$}
        \label{fig:ex_sigma}
    \end{subfigure}

    \caption{Examples of trajectories between $(R_2,L_2)$ and $(R_3,L_3)$ for different values of $s$ and of $\sigma$.  Subfigure a) : impact of $s$ on trajectories with $\sigma$ fixed. Subfigure b) : impact of $\sigma$ on trajectories with $s$ fixed. Each color corresponds to a value of the studied hyper parameter and for each of those values, one example of a simulated trajectory has been plotted.}
    \label{fig:ex}
\end{figure}
\begin{figure}[ht]
    \centering
    \begin{subfigure}[t]{0.49\textwidth}
        \centering
        \includegraphics[width=\linewidth]{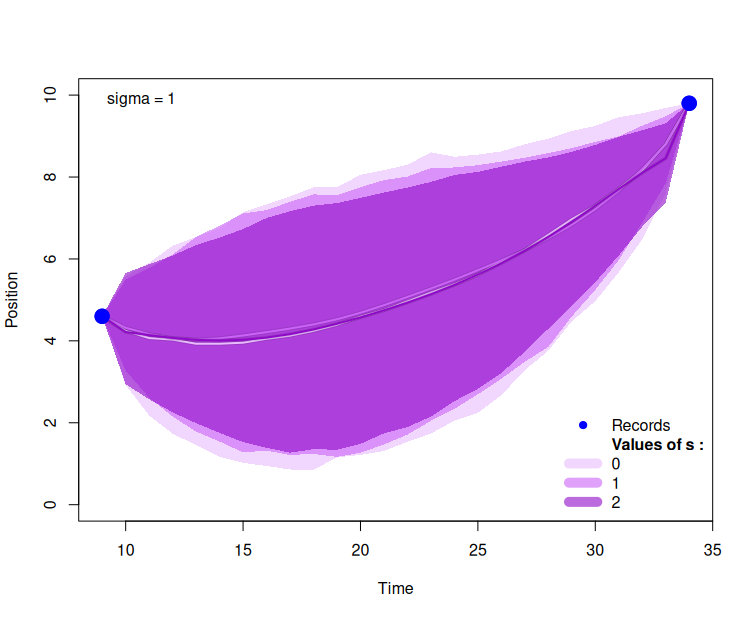}
        \caption{Influence of $s$ with $\sigma = 1$}
        \label{fig:s_gen}
    \end{subfigure}
    \hfill
    \begin{subfigure}[t]{0.49\textwidth}
        \centering
        \includegraphics[width=\linewidth]{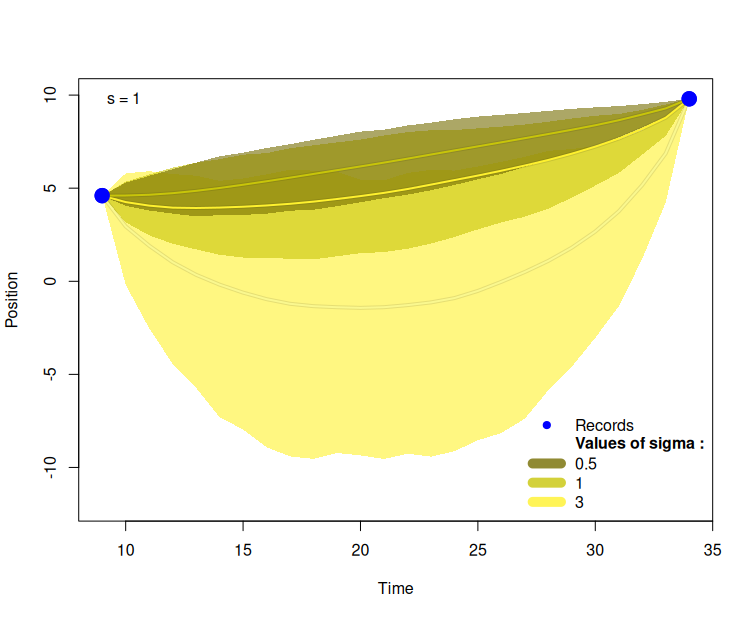}
        \caption{  Influence of $\sigma$ with $s = 1$}
        \label{fig:sigma_gen}
    \end{subfigure}

    \caption{Envelopes and averages of the trajectories between $(R_2,L_2)$ and $(R_3,L_3)$ for different values of $s$ and of $\sigma$. Subfigure a) : impact of $s$ on trajectories with $\sigma$ fixed. Subfigure b) : impact of $\sigma$ on trajectories with $s$ fixed. Areas of different colors are the envelopes (see equations \eqref{eq:quantile} and \eqref{eq:fct_dist}) of the trajectories simulated for different values of the studied parameter. }
    \label{fig:faisceaux_gen}
\end{figure}
\\

Figures \ref{fig:ex} and \ref{fig:faisceaux_gen} illustrate the impact of the hyper parameters $s$ and $\sigma$ on the trajectories. The effects depend in part on the range of the studied records and on the specified time step. We have chosen to retain the corresponding time step ($l=26$) and to ensure that the hyper parameters tested produce trajectories that are relatively close to the one we are attempting to reproduce (i.e. the one in Figure \ref{fig:intro}). Figure \ref{fig:ex} shows simulated trajectories between those records. For the left panel, the value of the parameter $\sigma$ was set to $1$ and for each value of $s$ in $\{0,1,2\}$, a trajectory was plotted. The higher the value of s, the more the trajectory will be smoothed. Once the different portions of the curve have been assembled, this parameter has also the effect of giving a more or less smooth appearance to the trajectories around the records. For the right panel, we did the same for $\sigma$ in $\{0.5,1,3\}$, fixing $s=1$. In general, the higher the value of the parameter $\sigma$, the greater the amplitude of the trajectories. Thus, for a given set of hyper parameters, the simulated trajectories vary to a greater or lesser extent, and giving just one example of trajectory as in Figure \ref{fig:ex} is not enough to have an overall view of the space simulations can cover. Therefore, Figure \ref{fig:faisceaux_gen} remedies this, it was produced by calculating the mean and the envelopes of $K=1000$ simulated trajectories, $X^1,...X^{K}$ , between these two records for each value of $s$ and $\sigma$. The envelopes correspond to the colored areas that represent the interval between the 5\% quantile series $Q^5$ and the 95\% quantile series $Q^{95}$. The quantile series ${Q^5_1,...,Q^5_l}$ and ${Q^{95}_1,...,Q^{95}_l}$ have been defined by computing for each timestep $t \in \{1,...,l\}$,
\begin{equation}\label{eq:quantile}
    Q^{\alpha}_t=\inf \{q,F_t(q)\geq \alpha \},
\end{equation} where $F_t(q)$ is the empirical distribution function for timestep $t$ defined by, 
\begin{equation}\label{eq:fct_dist}
    F_t(q)=\frac{1}{K}\sum_{k=1}^{K}\mathbb{1}(X^k_t\leq q)
\end{equation} and $\alpha \in \{0.05,0.95\}$. The lines are the mean $M_1,...,M_l$ of the trajectories and have been defined as,
\begin{equation}
    M_t=\frac{1}{K}\sum_{k=1}^{K}X_t^k.
\end{equation}
It shows that the higher the value of $s$, the smaller the range covered by the trajectories but the effects are barely noticeable, as this hyper parameter mainly affects the smoothness of the trajectories. Conversely, the higher the value of $\sigma$, the larger this range. Moreover, $s$ has an influence in the neighborhood of the records. The higher its value, the smoother the simulated trajectories are near the records. The larger $s$ is, the flatter the trajectories will be at record levels and the less pronounced the difference between two records will be, we can say that this parameter affects the horizontal amplitude of the trajectory beams. Conversely, $\sigma$ affects the vertical amplitude. The higher its value, the more the trajectories are allowed to descend between the two records.

\begin{figure}
\centering
    \includegraphics[width=0.6\linewidth]{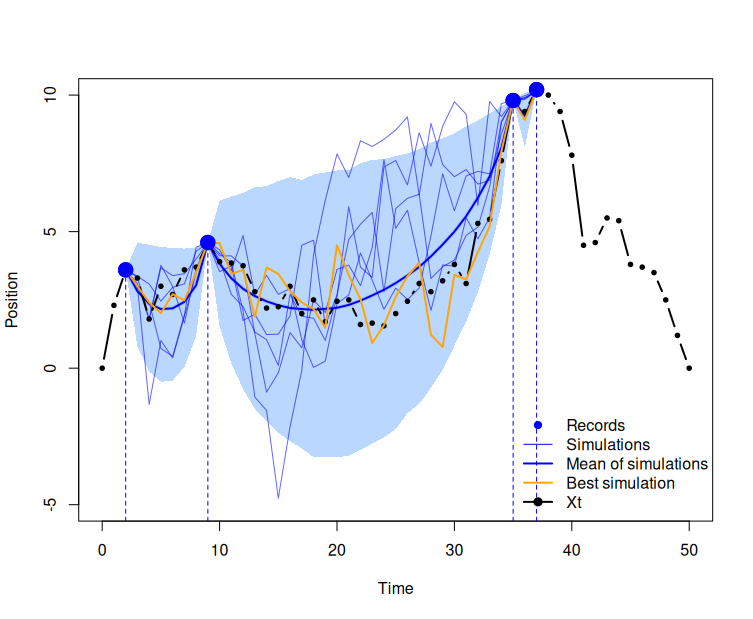}
    \caption{Examples, envelope and mean of simulations between the records of Figure \ref{intro}. Blue area is the envelope (see equations \eqref{eq:quantile} and \eqref{eq:fct_dist}). Blue dots are the records of the series $X_t$ which was presented in Figure \ref{intro} and plotted in black here. The thin blue lines are 10 examples of simulations made for $s=0$ (no smoothing) and $\sigma=1.7$ and the thick blue curve is the mean of 10 000 simulations of trajectories with the same set of hyper parameters. The orange line is the best simulation according to the MSE (see Equation \ref{mse}). }
    \label{fig:simu_sur_intro}
\end{figure}

Figure \ref{fig:simu_sur_intro} shows an example of simulations carried out between the data points of the $X_t$ series presented in Figure \ref{fig:intro}. This series contained four records, so the simulation algorithm simulates a trajectory only between them, i.e. over three intervals delimited by the dotted vertical lines. The hyper parameters were chosen arbitrarily, ensuring that the average of the simulated trajectories (thick blue curve) closely matched the variations in $X_t$. In total, $10,000$ trajectories were simulated using these hyper parameters, enabling the envelope (see equations \eqref{eq:quantile} and \eqref{eq:fct_dist}) shown in blue to be determined. Of these $10,000$ simulations, $10$ were selected at random and are shown in the figure as the one which was the closer to $X_t$ in terms of MSE (see Equation \ref{mse}). We note that between the data points $(R_3,L_3)$ and $(R_4,L_4)$ there is only one time step, so the simulated trajectories are very close to one another. The purpose of this figure is to illustrate the simulations that can be generated using our stochastic generator. However, it is important to note that, as Figure \ref{fig:comparaison_faisceaux} also shows, whilst the hyper parameters allow for a wide variety of trajectories, they can be selected only based on observations. In the next part, we will see how these hyper parameters can be determined if additional observations are available alongside the historical records.

\section{Application to glaciers retreat}\label{application}
\subsection{Context}\label{subsec3}

This method of simulating trajectories conditionally on records can be useful in several fields. In this section, we explore how it can be applied to address a glaciological problem. Glaciers make an important contribution to mountain landscapes \citep{mercier2024atlas} and are the most visible indicators of global warming  \cite[see, e.g.,][]{hock2019high,gorin2024recent}, with the influence of human activities on climate conditions resulting in their accelerated melting and retreat (\cite{theglambieteamCommunityEstimateGlobal2025}). These changes result in local disaster risk \citep{buntgen20252025} and more broadly affect water resources \citep{beniston2014assessing} and freshwater ecosystems \citep{cauvy2019global} in mountain environment and significantly contribute to sea level rise worldwide \citep{zemp2019global}. Understanding the long-term evolution of glaciers is therfore a real challenge and a work of interest. \cite{roeWhatGlaciersTell2011} explains that a glacier is defined by two zones, one of accumulation, the highest where the snow is accumulated, and one of ablation, constituting the lower part which may melt. The equilibrium line defines the meeting point between areas of accumulation and ablation. The mass balance of a glacier corresponds to the difference between the mass it has gained and the mass it has lost over a period of time. It directly links the climate and the geometry of the glacier, i.e. its thickness and length. However, the link between variations in the length of a glacier and climate disturbances is not straightforward. Glaciers can react with varying degrees of delay, known as response time. The response time of the glacier to a climatic perturbation varies from a few years to centuries and the time lag between the climate forcing and the reaction of the front depends on divers parameters like the lenght of the glacier, the thickness of the ice the the slope of the valley. Namely, for the same climate signal, different glaciers react differently. 

A glacier moves along a section of the valley, and the ice flow makes that the frontal position of a glacier always moves along the same axis of the valley. Gravity causes the ice that makes up the glacier to descend along this axis, but according to the mass balance, it is possible for the glacier's front and its equilibrium line to rise along the axis. During a period of positive mass balance, the frontal position of the glacier descends along this axis. Periods of descent may be interspersed with periods of rest or ascent. The length of these periods depend on climate fluctuations and glacier characterics. As the glacier ice flow downwards, it accumulates sediment and rock at its front, forming a bulge called the terminal moraine \citep{jomelli2024synthesis}. When the glacier enters a phase of retreat, the moraine it formed on the way down remains in place. The next time the glacier advances, two scenarios are possible. Either the glacier descends lower than before, destroying the moraine it left behind and forming a new, lower moraine, so that the first is no longer visible. Or it begins its retreat before crushing the previous moraine, forming a second, higher one, and both are still visible. Front variations correspond to the difference in meters of the position of the glacier's front between two given years. They can directly be linked to the positions of the glacier at time $t$, in our representation defined above, which will be noted $X_t$ in the following.

\begin{center}
    \includegraphics[width=0.9\linewidth]{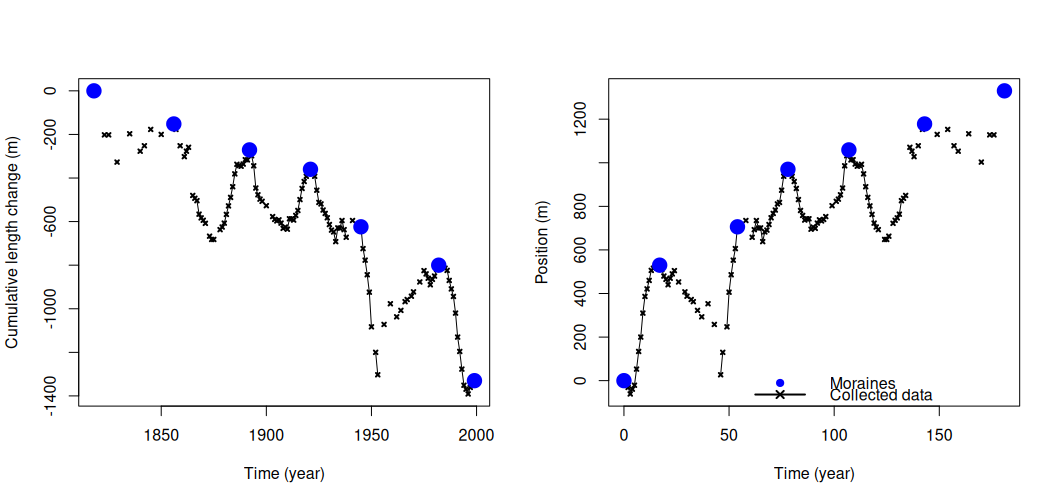}
    \captionof{figure}{Transformation of the observed \textit{Glacier des Bossons} trajectory from 1818 to 1999 to fit the record theory framework. Source: \cite{nussbaumer2012little}. The blue dots are the observed moraines, and the black line is the real observed trajectory of the frontal position of the \textit{Glacier des Bossons}. Both are measured along a longitudial axis corresponding to the glacier flow direction. Subfigure a) : the x-axis represents the years and the y-axis the frontal position relative to the first moraine position that corresponds to the glacier maximal extent during the Littel Ice Age. Subfigure b) : the x-axis represents the number of years between 1999 and the year plotted and the y-axis position relative to the 1999 glacier frontal position.}
\label{fig:plot_mor}
\end{center}

\subsection{\textit{Glacier des Bossons} case-study data}\
The data available for the \textit{Glacier des Bossons} front trajectory is shown in Figure \ref{fig:plot_mor}. The \textit{Glacier des Bossons} is located in the French Alps, in Haute-Savoie. It is a prime example of a well-known glacier for which time-frontal position pairs have been recorded for almost 200 years on a more or less regular basis. This glacier is characterized by its rapid dynamics, namely the ice flows quite quickly and so that the peaks in the trajectory of the frontal position are particularly pronounced \citep{nussbaumer2012little}. This dynamics allowed several phases of advance and retreat over the period studied, which corresponds to the non-monotonic warming phase that occurred since the end of the Little Ice Age, the coldest period of the last millenia \cite[see, e.g.,][]{grove2019little,solomina2015holocene,solomina2016glacier}. In the following, we will focus on the observations made between 1818 and 1999 ($\mathcal{T}=182$), this represents 8 moraines spaced 30 years apart and 222 meters, on average. Before 1862, there is uncertainty regarding the measurements of the glacier frontal position, which vary between 50m and 100m depending on the year. After 1862, the positions are given without uncertainty. The amplitude of the glacier's position on the longitudinal axis is 1,330m as written in the Table \ref{tab:moraines_bossons}. The fact that measurements have only been available for a period of around 200 years is not a problem, as this period is sufficient to track the overall massive retreat of the glacier since the end of the Little Ice Age. Therefore, by limiting ourselves to this period, we can consider that the reference position at 1,330 metres represents a true maximum. On right pannel of the Figure \ref{fig:plot_mor}, we can notice the non-stationarity of this series. Moraines appear to have a frequency greater than $\frac{1}{t}$ where $t$ denotes the year and glacier position seems to have a negative trend. It seems easy enough to link the notion of moraine to that of record. Their position are decreasing and are located on local maxima as defined on Section \ref{intro}. The  right pannel of Figure \ref{fig:plot_mor} shows how the data from the left pannel are used to build a new representation where the moraines are increasing. The x-axis now represents the number of years that have passed since the last moraine was created. Thus, the first blue dot corresponds to the last row of the Table \ref{tab:moraines_bossons}. The y-axis represents the relative position. It is closer to an altitude and corresponds precisely to the distance in meters separating the glacier from its lowest position ever recorded (the rightmost point of the black curve in the left pannel of Figure \ref{fig:plot_mor}).
In the following, the visual representation of the right pannel of Figure \ref{fig:plot_mor} will be used. Moraines correspond to the records of the variable representing the local maxima of the time series of glacier front positions after axis inversion. Noting $V_1,...,V_n$ the series of glacier positions in the left representation of Figure \ref{fig:plot_mor} and $X_1,...,X_n$ the series of glacier positions in that of the right, the link between the two is as $\forall t \in  \{1,...,\mathcal{T}\}, X_t=V_{max}-V_{\mathcal{T}-t}$ where $V_{max} = \max(V_1,...,V_\mathcal{T})$. Then, moraines are exactly defined by the subset $\{R_n\}_{n\geq1} =\{X_t> \max(X_1,...,X_{t-1},X_{t+1})),t\geq 1 \}$.

\begin{table}[h]
    \centering
    \begin{tabular}{|ccc|ccc|}
        \hline
         \multicolumn{3}{|c}{Left pannel of Figure \ref{fig:plot_mor}}&  \multicolumn{3}{|c|}{Right pannel of Figure \ref{fig:plot_mor}}\\
        \hline
        \textbf{Apparition order}&\textbf{Year}&\textbf{Cumulative length change}&\textbf{Record index}&\textbf{Timestep}&\textbf{Relative position}\\
        \hline
       $ t$&$L_t$&$V_t$& n &$L_n=L_{t_{max}}-L_t$&$R_n=X_t=V_{max}-V_{\mathcal{T}-t}$\\
        \hline
         1&1818&0 m&7&181&1330 m \\
         2&1856&-152 m &6&143&800 m\\
         3&1892&-271 m&5&107&624 m\\
         4&1921&-360 m&4&78&360 m\\
         5&1945&-624 m&3&54&271 m\\
         6&1982&-800 m&2&17&152 m\\
         7&1999&-1330 m&1&0&0 m\\

        \hline
    \end{tabular}
        \caption{Moraines of the \textit{Glacier des Bossons}.  The first moraine of the \textit{Glacier des Bossons} series defines the reference position. It was created on the year 1818 which correspond, for the second representation, to the last record, observed at timestep 181. $V_{max}=0$m and $L_{t_{max}}=1999$ and $\mathcal{T}=182$. }
    \label{tab:moraines_bossons}
\end{table}

\subsection{Choice of the hyper parameters in the \textit{Glacier des Bossons} case}\label{hyper parameters}
\begin{figure}[ht]
    \centering
    \begin{subfigure}[t]{0.6\textwidth}
        \centering
        \includegraphics[width=\linewidth]{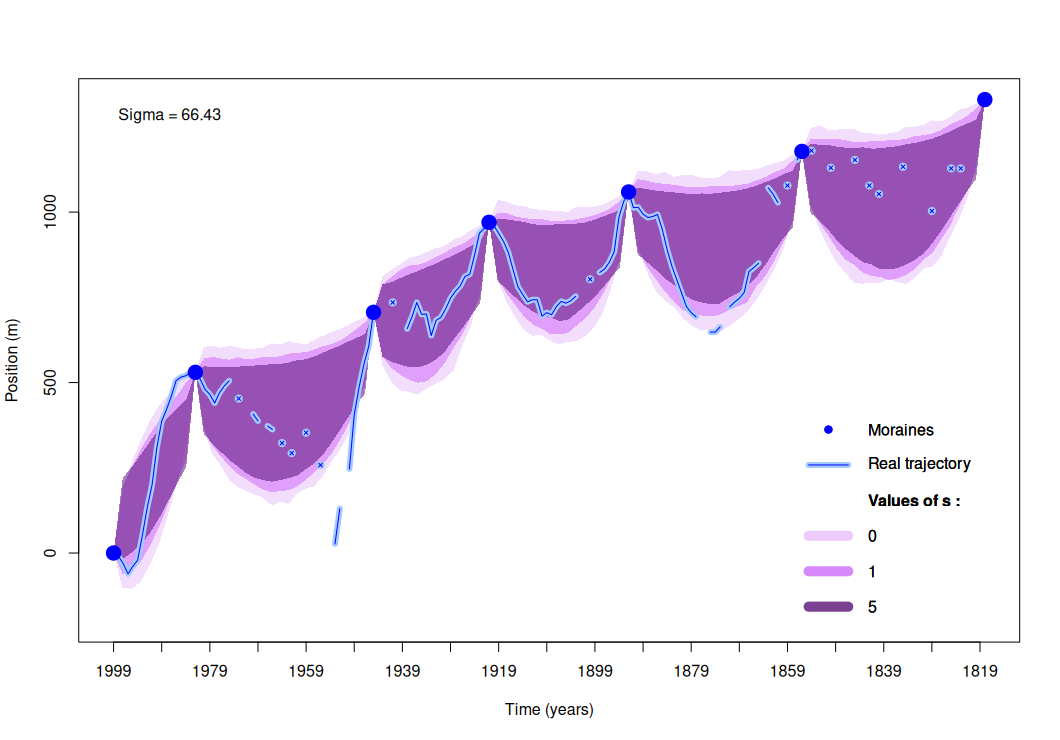}
        \caption{  Influence of $s$ with $\sigma=66.43$ }
        \label{fig:sd_faisceaux}
    \end{subfigure}
    \hfill
    \\
    \begin{subfigure}[t]{0.6\textwidth}
        \centering
        \includegraphics[width=\linewidth]{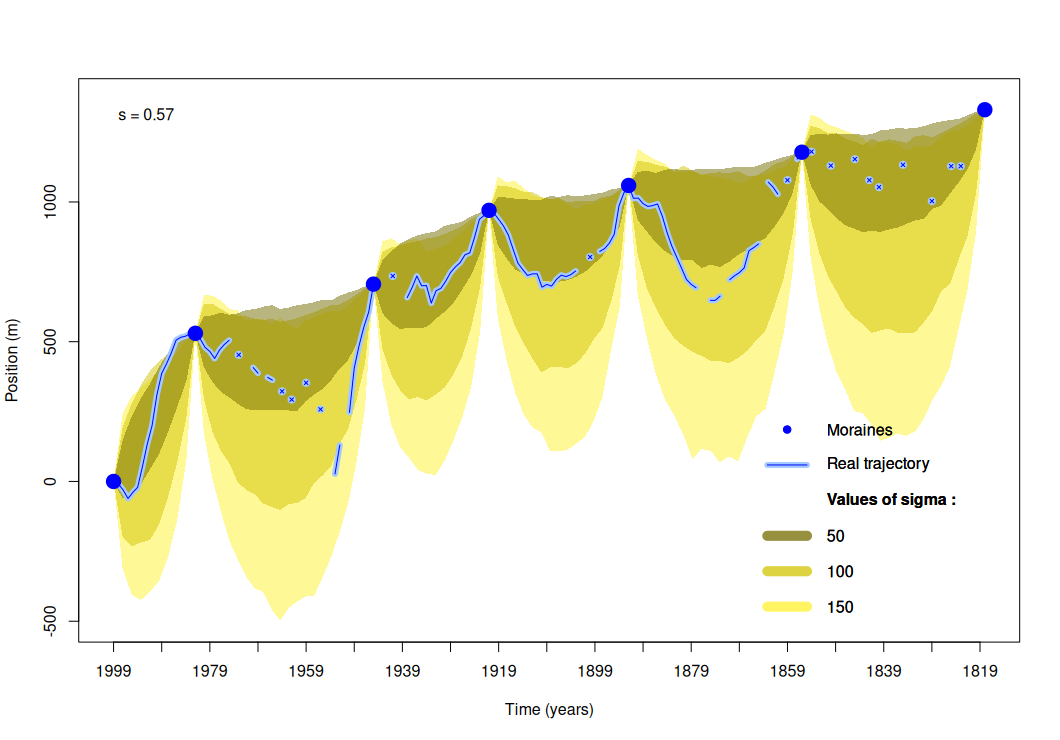}
        \caption{ Influence of $\sigma$ with $s=0.57$}
        \label{fig:ssm_faisceaux}
    \end{subfigure}

    \caption{Envelops of trajectories for different values of $s$ and of $\sigma$.  Subfigure a) : $s$ impact on trajectories with $\sigma$ fixed. Subfigure b) :  $\sigma$ impact on trajectories with $s$ fixed. The blue line is the real \textit{Glacier des Bossons} trajectory. Blue dots are the observed moraines. Areas of different colors are the envelops (see equations \eqref{eq:quantile} and \eqref{eq:fct_dist}) of the trajectories simulated for different values of the studied parameter. When $s$ is fixed at $0.57$, an increasing $\sigma$ causes the area covered by the trajectories to widen downwards.}
    \label{fig:comparaison_faisceaux}
\end{figure}

Applying our approach to the \textit{Glacier des Bossons}, Figure \ref{fig:comparaison_faisceaux}, as Figure \ref{fig:faisceaux_gen}, illustrates the influence of the Brownian generator hyper parameters. It was obtained by fixing respectively the values of $s$ and $\sigma$ and varying the second parameter within the ranges defined by their priors. We did not have much prior information about the laws governing these parameters, so we opted for uniform priors and tested the chosen value ranges by studying the sensitivity of the priors. The ranges of values that are the best for our application are, for $\sigma$, between $50.0$ and $ 150.0$, and, for $s$, between $ 0.0$ and $5.0$, they were informed by exploratory simulations, ensuring that the hyper parameter space was restricted to regions producing physically meaningful outputs. The fixed hyper parameter values were chosen based on the study conducted on data from the \textit{Glacier des Bossons}, whose results are presented in Table \ref{tab:results_param_bossons}. For each of the four values considered, 1000 trajectories were simulated using the Brownian generator applied to the moraines of the \textit{Glacier des Bossons} detailed above Table\ref{tab:moraines_bossons}. The mean of these trajectories is shown as a solid line, and the envelops is represented by the shaded area and have been computed as explained in equations \ref{eq:quantile} and \ref{eq:fct_dist}. The studied values are $s\in\{0,1,5\} $ ($s=0$ corresponds to no smoothing). As shown in Figure \ref{fig:sd_faisceaux}, the higher its value, the smoother the trajectories and the more the bundles tighten horizontally. Variations of this parameter do not affect the vertical amplitude of the bundles, unlike those of $\sigma$, as shown in Figure \ref{fig:ssm_faisceaux}. The studied values are $\sigma\in\{50,100,150\} $.\\

The smoothing function applied to the trajectories between each pair of records, see Equation \eqref{eq:w_s},  prevents the   calculation of a likelihood.   \citep{lueckmannBenchmarkingSimulationBasedInference2021}.
Thus, a likelihood-free approach is needed to choose the hyper parameters $(\sigma,s)^T$. One key aspect is that, given a pair $(\sigma,s)^T$ and records data like $\mathcal{M}=\{(R_n,L_n),n=1,...,N\}$, Algorithm \ref{alg1} described in Section \ref{model} can  rapidly  produced numerous random trajectories. In addition, we have a series of observations for the \textit{Glacier des Bossons} that can be used. Hence, techniques based on simulation schemes could help us to choose the best hyper parameters in our context. For example, approximate Bayesian computation (ABC) methods allow observations and simulations to be compared using summary statistics. As explained in  \cite{FrontierSimulationbasedInferencea}, ABC methods require a threshold choice  to select satisfactory simulation statistics. 
Furthermore, the effectiveness of these ABC schemes  depends on the statistics chosen, and this choice is not always easy to make. To bypass these steps, other simulation methods based on the use of machine learning has been developed in recent years \cite[see, e.g.,][]{papamakarios2018fastepsilonfreeinferencesimulation,Papamakarios2021,sainsbury2024neural}. Neural Bayes estimators (NBE) constitute a flexible class of simulation-based estimators that approximate Bayes estimators under a specified loss function by training neural networks on synthetic datasets generated under prior sampling \citep{papamakarios2018fastepsilonfreeinferencesimulation}. These approaches  allow optimal parameters selection without having to choose statistics or a threshold. 
The objective of this section is to introduce a simulation-based procedure for identifying hyper parameters of our Brownian model that generate trajectories consistent with an observed physical trajectory. The relationship between simulated trajectories and their associated model parameters is learned using a feedforward neural network (MLP – Multi-Layer Perceptron) designed for parameter selection from simulated observations (see, e.g., \cite{chollet2022deep}). In the following it will be referred to as the Neural Network Model (NNM). It was coded in R using the libraries keras and tensorflow  \cite[see, e.g.,][]{kapoor2022deep,abadi2016tensorflow}. The architecture was chosen for its simplicity, universal approximation capability, and robustness when dealing with limited datasets. The Brownian simulation model defines a generative relationship $X \sim p(Z \mid \mathbf{p}) $ between parameters and trajectories, where $\mathbf{p}=(\sigma,s)$ denotes the hyper parameters of the Brownian model and $Z \in \mathbb{R}^{\mathcal{T}}$ is a vector of trajectory features. The NNM is trained to approximate the inverse relationship by learning a mapping $\hat{\mathbf{p}} = f_\theta(X)$ from simulated trajectories to their associated hyper parameters (\cite{370fbeadb5584ba9ab2938431fc4f140}). Here, $\hat{\mathbf{p}} = (\hat{\sigma}, \hat{s})$ is a point estimate of the conditional distribution $p(\mathbf{p} \mid X)$ and the parameters of $f$, $\theta$ are optimized by minimizing the expected gap between true simulation parameters and network predictions by computing  
\begin{equation}
    \theta^\star = \arg\min_\theta \; \mathbb{E}_{\mathbf{p}, X \sim p(X|\mathbf{p})}
\left[ \ell\big(\mathbf{p}, f_\theta(X)\big) \right],
\end{equation}
where $\ell$ is the loss function representing the  Mean Squared Error (MSE) (see Equation \eqref{mse}) between the estimated and target parameters and $\mathbb{E}$ is approximated thanks to the simulations. Optimization is performed with an initial learning rate of $10^{-3}$. Training is conducted on a synthetic dataset split into 80\% for training and 20\% for validation, with early stopping based on the validation performance. By minimizing the loss function between true simulation parameters and their predictions, the NNM learns to associate trajectory features with parameter values, without providing an explicit or interpretable distance measure between trajectories. It approximates the conditional expectation of the parameters given the trajectory representation. Once trained, the model is applied to the observed trajectory $\{X^{obs}_t\}_{1\leq t\leq \mathcal{T}}$ to infer parameter values $\hat{\mathbf{p}}_{\text{obs}} = f_{\theta^*}(\{X^{obs}_t\}_{1\leq t\leq \mathcal{T}})$. These parameters are then used to generate new trajectories from the Brownian model, which are compared to the observed trajectory using MSE (see Equation \eqref{mse}) and MAPE (see Equation \eqref{mape}) metrics (\cite{ABC}). As a result, the parameters represent a point estimate induced by the training data and model architecture rather than a unique or physically interpretable solution. This method does not aim to recover the true parameters underlying the observed trajectory, but rather to restrict the parameter space to regions that are generatively consistent with the observation.

\begin{center}
    \includegraphics[width=0.7\linewidth]{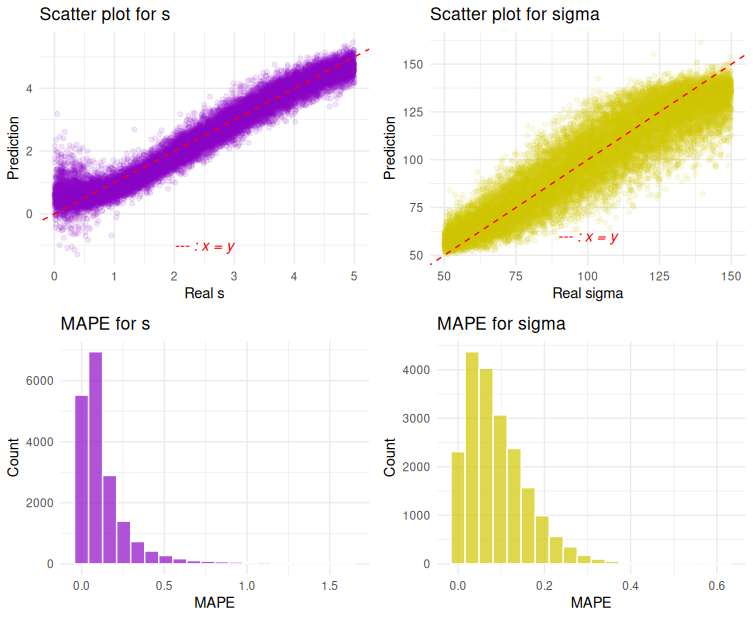}
    \captionof{figure}{Results of the NBI predictions for the test sample (2 000 simulations). First line : scatter plots of the predicted parameters against the real ones. Second line : distribution of the MAPE (see Equation \ref{mape}). First column : parameter $s$. Second column : parameter $\sigma$. }
    \label{fig:scatters}
\end{center}

Figure \ref{fig:scatters} illustrates the prediction results produced by the neural network model trained on the test sample. The top plots show scatter plots of the true values of the parameters $s$ and $\sigma$ against the values predicted by the NN model. \\
The model reproduces the parameter $s$ better than $\sigma $ according to the coefficients of determination ($R^2=0.95$ for $s$ and $R^2=0.85$ for $\sigma$). It slightly overestimates the values of $s$ when they are around $0$.

\begin{center}
\begin{table}%
\centering
\begin{tabular*}{500pt}{@{\extracolsep\fill}lll|lll|ll@{\extracolsep\fill}}

\multicolumn{3}{c}{\textbf{Learning results}} &
\multicolumn{3}{c}{\textbf{Comparisons}} &
\multicolumn{2}{c}{\textbf{Predictions}} \\
\midrule

\textbf{$R^2_{s_B}$} &\textbf{ $R^2_{\sigma_B}$ }& \textbf{$R^2_{mean}$  }&
\textbf{Area ratio }&\textbf{MSE ratio}& \textbf{MAPE ratio}&
\textbf{$ s$ }&\textbf{$ \sigma$ }\\

\midrule
0.95&0.85 & 0.90  & 0.66  & 0.55 &0.60&0.57& 66.43\\
\bottomrule

\end{tabular*}
\caption{Estimations of the hyper parameters for the \textit{Glacier des Bossons}. The three first columns are the coefficients of determination of the linear regression between real values and predictions. The area of the envelops (see equations \eqref{eq:quantile} and \eqref{eq:fct_dist}) of the trajectories simulations with the predicted parameters represents 66\% of the one with random parameters. The MSE between the real trajectory and the trajectories simulations with the predicted parameters represents just 55\% of the one with random parameters. Last two columns are the results of the model prediction with the trajectory of the \textit{Glacier des Bossons} as input.}
\label{tab:results_param_bossons}
\end{table}
\end{center}

\begin{center}
    \includegraphics[width=0.7\linewidth]{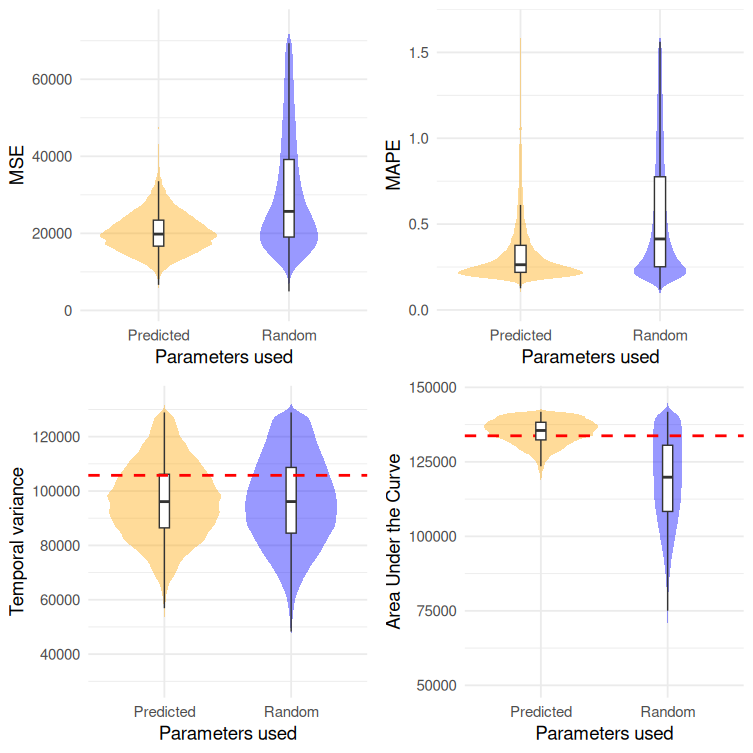}
    \captionof{figure}{Violin plots of the distributions of the 95\% interval of four metrics with the predicted parameters (orange) or with random ones (blue). First line illustrates the errors, MSE  (see Equation \eqref{mse}) and the MAPE (see Equation \eqref{mape}). For both subfigures of the first line, the violin plots show the distribution of the error between the real \textit{Glacier des Bossons} trajectory and the predictions. The left violin plots are made on simulations made with random parameters $\sigma$ and $s$ (chosen according to their prior distribution) and the right violin plots are made on simulations made with the parameters predicted by the neural network (presented in Table \ref{tab:results_param_bossons}).Second line illustrates the distributions of the 95\% interval of two different metrics (see equations \eqref{var}, \eqref{auc}) computed on trajectories made with the predicted parameters or with random ones. Red lines represent the values of those metrics for the interpolated real trajectory presented in Figure \ref{fig:interpolation}. }
    \label{fig:4metrics}
\end{center}

\subsection{Application of the final model to the \textit{Glacier des Bossons} case}\label{application }
Once $s$ and $\sigma$ chosen, they are used to generate 10,000 trajectories from the Brownian model. 
These trajectories are then compared to the real trajectory of the \textit{Glacier des Bossons} on the period over which the latter is avalaible using error metrics such as the MSE (see Equation \eqref{mse}) and the MAPE (see Equation \eqref{mape}).
As a baseline, we generate 10,000 trajectories using randomly sampled parameters from the prior parameter space and evaluate them using the same metrics. This comparison does not constitute a validation of parameter estimation accuracy, as the true parameters underlying the trajectory of the \textit{Glacier des Bossons} are unknown.
Instead, it provides an indirect assessment of the generative consistency of the inferred parameters, i.e., their ability to produce trajectories that are closer to the observed data than those obtained from uninformed parameter sampling. The results are presented in Figure \ref{fig:4metrics}. For both first metrics (MSE and MAPE), the trajectories simulated with random parameters exhibit higher average errors. A lower error for trajectories generated from chosen parameters is expected, as the neural network model is explicitly trained to map trajectory features to compatible regions of the parameter space. The random-parameter baseline should be interpreted as a lower bound on generative performance rather than as a competing inference method. Because the simulator is stochastic, the reported errors are averaged over multiple realizations for each parameter set. Figure \ref{fig:4metrics} also compares, using two other metrics, the trajectories produced using the chosen parameters with those produced using random parameters. It also compares those metrics repartition with their values computed on a linear interpolation of the \textit{Glacier des Bossons} trajectory. It has been done with the \textit{na.approx} function of the \textit{zoo} package of R and the result is presented in Figure \ref{fig:interpolation} . The first is temporal variance (see Equation \eqref{var}), and there is little difference. The next metric is the area under the trajectory curve (see Equation \eqref{auc}). The chosen parameters give it a higher average value and greatly reduce the standard deviation. Then, we can say that choosing the hyper parameters has a significant impact on simulated trajectories and allows us to refine the range of possible trajectories.

\begin{center}
    \includegraphics[width=0.7\linewidth]{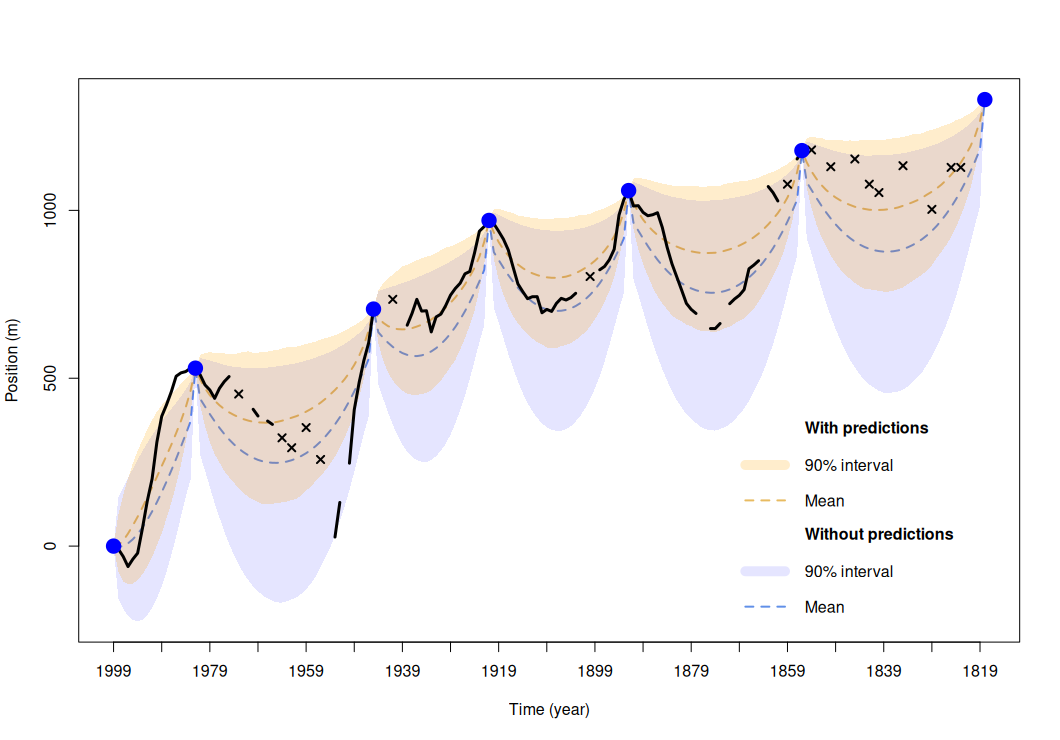}
    \captionof{figure}{ Area and mean of the simulated trajectories with or without the predicted parameters. The black line is the real trajectory of the \textit{Glacier des Bossons}. Blue elements correspond to simulations made with random parameters $\sigma$ and $s$ (chosen according to their prior distribution). Orange elements correspond to simulations made with the parameters predicted by the neural network (presented in table \ref{tab:results_param_bossons}). Transparent areas are the envelopes (see equations \eqref{eq:quantile} and \eqref{eq:fct_dist})  of the simulations and the dotted lines are the mean of the simulations.}
    \label{fig:results_faisceaux1}
\end{center}

\begin{center}
    \includegraphics[width=0.7\linewidth]{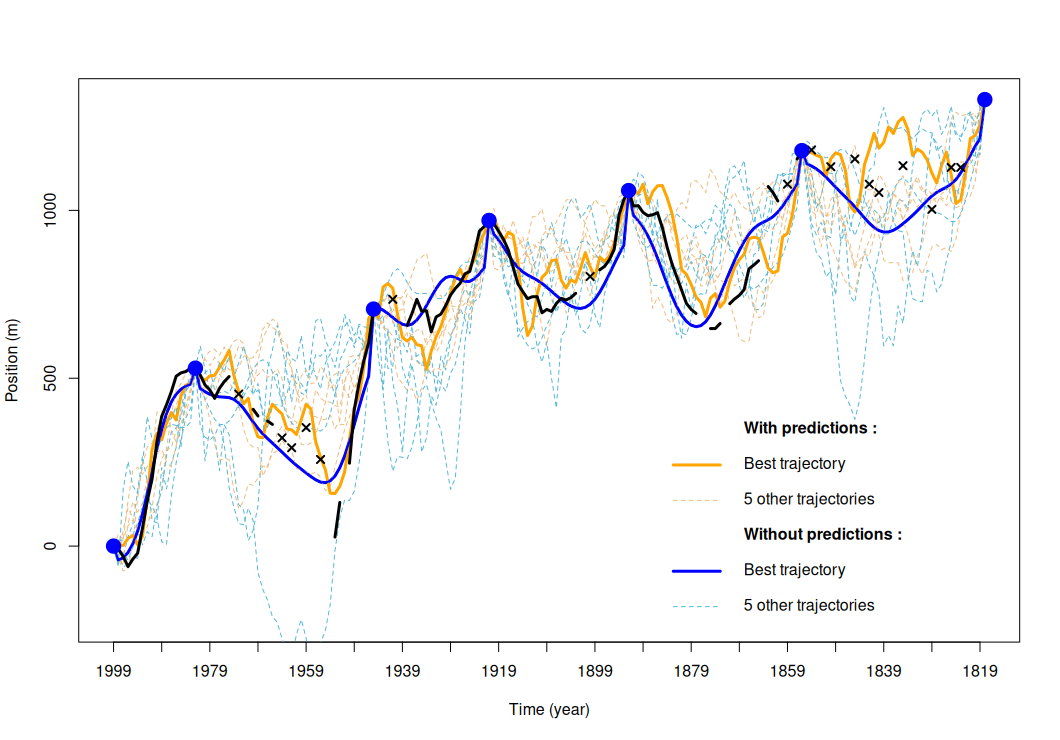}
    \captionof{figure}{Simulated trajectories with or without the predicted parameters. The black line is the real trajectory of the \textit{Glacier des Bossons}. Blue elements correspond to simulations made with random parameters $\sigma$ and $s$ (chosen according to their distribution). Orange elements correspond to simulations made with the parameters predicted by the neural network (presented in table \ref{tab:results_param_bossons}). Ten random trajectories for each parameters area have been represented as well as the best trajectory in term of MSE (see Equation \ref{mse}).}
    \label{fig:results_faisceaux2}
\end{center}

Figures \ref{fig:results_faisceaux1} and \ref{fig:results_faisceaux2} illustrate the simulation results obtained using the stochastic Brownian generator, with as input the moraines of the \textit{Glacier des Bossons} (represented with the blue dots) and, respectively, parameters chosen randomly or parameters predicted by the hyper parameters selection method.
To obtain these graphs, two sets of 1,000 simulations were used. The first, represented in blue, used pairs of parameters drawn from their uniform distributions. The second, in orange, used the parameters predicted by the neural network. In Figure \ref{fig:results_faisceaux1}, the shaded areas correspond to the envelops of the trajectories, and as indicated in Table \ref{tab:results_param_bossons}, the prediction of parameters reduces the width of the bundle by almost a third. Between moraines 2 and 3 (corresponding respectively to years 29 and 70 according to Table \ref{tab:moraines_bossons}), the envelopes of the trajectories with fixed parameters does not fully contain the observed trajectory. Figure \ref{fig:results_faisceaux2} provides an overview of the trajectories that can be simulated by the Brownian model. The trajectories shown in orange are, on average, closer to reality than those shown in blue and admit a variability more consistent with the observations. 

\section{Conclusion and outlooks}\label{ccl}

This work is part of the statistical study of records and offers a new approach to simulate continuous time series, that may deviate from the stationary framework, conditional on a set of specific records. The proposed Brownian method has the Markovian advantage of dividing the simulation problem into several independant ones, between the different records. As illustrated with our case study, the model we arrive at is fast, reliable and, thanks to the hyper parameters, offers great flexibility in the trajectories it is capable of simulating, as well as in the type of the records studied. This offers great freedom in terms of areas of application, namely the various problems for which record series is the sole available data. Many avenues also exist for improving the model, for instance, extending it to a multivariate setting to consider several series of records altogether or accounting for uncertainty in dating and/or missing records.\\
Regarding the application, our stochastic generator provides a statistically rigorous framework for reconstructing glacier front trajectories from sparse moraine positions, which are often the only available proxy for long-term (centennial to millennial) glacier dynamics. Unlike traditional conditional simulation methods, our approach explicitly enforces the record-ordering constraint while accommodating non-stationary behavior, making it particularly suitable for glaciers exhibiting complex advance-retreat patterns such as the \textit{Glacier des Bossons}. A key advantage of our approach including NB inference is its ability to fill gaps by inferring hyper parameters from partially observed trajectories (e.g., 19th–20th century measurements). We thus generate plausible trajectories for unmonitored periods, highlighting its potential for the numerous glacier situations for which partial trajectories have been monitored.
However, the method’s dependence on hyper parameters introduces practical challenges. In many real-world cases, continuous front position records are unavailable, and only moraine-based records exist. While our NB inference method effectively estimates hyper parameters from fragmentary observations, extending its applicability to glaciers with only moraine records requires further development. Two promising directions include: (i) categorizing hyper parameters based on glacier-specific attributes (e.g., slope, altitude, ice volume) using a multi-glacier dataset, and (ii) the multivariate extension of the model to jointly analyze glaciers with full trajectories and those with only moraines, leveraging spatial dependencies among neighboring glaciers. Future work could also integrate physical glacier models like Elmer/Ice \cite[see, e.g.][] {gagliardini2013capabilities, https://doi.org/10.1029/2022GL100363} to validate simulations in unobserved contexts and refine hyperparameter ranges for different glacier types.\\
Reconstructing long-term glacier front fluctuations is critical for understanding glacier-climate interactions, including response times to external forcings (e.g., post-Little Ice Age warming). While physical models (e.g., Elmer/Ice) simulate ice dynamics with high precision, their accuracy relies on past climatic forcings, which are often poorly constrained for pre-instrumental periods. Our stochastic approach, independent of climatic inputs, offers a complementary tool by generating trajectories solely from geomorphological records (moraines). This may provide different  benefits: i) Benchmarking physical models by providing stochastic trajectories to validate or challenge physical model outputs in periods lacking climatic data (e.g., pre-1850); ii) Identifying discrepancies between stochastic and physical simulations that may reveal internal glacier dynamics or climatic forcings not accounted for; iii) Analyzing glacier-climate sensitivity. Ultimately, combining stochastic and physical models may help disentangle climatic influences from intrinsic glacier properties, to, e.g., improve projections of glacier contributions to sea-level rise.

\section*{Acknowledgments}
This research was supported by the French National Research Agency under the grants IRIMONT (ANR-22-EXIR-0003), CIME (ANR-25-CE01-4408-01) and SICIM (ANR-24-EXMA-0012). The authors are grateful to S. Nussbaumer for sharing the \textit{Glacier des Bossons} data set and for fruitful related discussions. N.E. is member of the Grenoble Risk Institute (https://risk.univ-grenoble-alpes.fr/en) and OSUG. \\
Part of Naveau’s research work was supported by the   Agence Nationale de la Recherche via:  the SICIM and SHARE PEPR Maths-Vives project (France 2030 ANR-24-EXMA-0008), EXSTA grant (ANR-23-CE40-0009-01), PORC-EPIC, the PEPR TRACCS program  (PC4 EXTENDING, ANR-22-EXTR-0005), and  the PEPR   IRIMONT (France 2030 ANR-22-EXIR-0003). He has also benefited  from the Geolearning research chair.u

\subsection*{Author contributions}
PN and NE designed the research; NE and VJ acquired the funding; MM, MP, and PN designed the stochastic generator; MM ran the model on the simulated and real data;  All authors contributed to the analysis and interpretation of results; MM drafted the paper. All authors edited the paper. 

\appendix

\section{Properties of the Vervaat transform}\label{appen:vervaat}

\begin{proposition}\label{prop:vervaat}
    Let $\lbrace F_t\rbrace_{t\in\interv{1}{\mathcal{T}}}$ be a time series such that $F_1<F_{\mathcal{T}}$. Let $\lbrace V(F)_t \rbrace_{t\in\interv{1}{\mathcal{T}}}$ be its  Vervaat transform defined in~\eqref{eq:vervaat}. 
    If $T_{max}$ in~\eqref{eq:vervaat} is such that $T_{max}<\mathcal{T}$, then
    \begin{itemize}
        \item[(i)] It has the same endpoints as the time series $\lbrace F_t\rbrace_{t\in\interv{1}{\mathcal{T}}}$, i.e.  $V(F)_1 = F_1$ and $V(F)_{\mathcal{T}} = F_\mathcal{T}$
        \item[(ii)] Its endpoints are local maxima, i.e. $V(F)_{2}\le V(F)_{1}$ and $V(F)_{\mathcal{T}-1}\le V(F)_{\mathcal{T}}$
        \item[(iii)] Its ending point $V(F)_{\mathcal{T}}$ is also a global maximum, i.e. $V(F)_{\mathcal{T}}\ge \max\limits_{1\le t\le \mathcal{T}-1} V(F)_t$
    \end{itemize}
\end{proposition}

\begin{proof}
Following~\eqref{eq:vervaat}, we have
\begin{align*}
V(F)_{1}=F_{1}-F_{T_{max}}+F_{T_{max}-1+1}
=F_{1}
\end{align*}
and 
\begin{align*}
V(F)_{\mathcal{T}}=F_{\mathcal{T}}-F_{T_{max}}+F_{T_{max}-\mathcal{T}+\mathcal{T}}
=F_{\mathcal{T}}
\end{align*}
Both equalities prove Point (i). Note then that, by definition of $T_{max}$ in~\eqref{eq:vervaat}, we have
\begin{align*}
    V(F)_{2}
    =\begin{cases}
    F_{1}-F_{T_{max}}+F_{T_{max}-1+2} 
    =F_{1}+(F_{T_{max}+1}-F_{T_{max}})\leq F_{1}=V(F)_{1} & \text{if } T_{max}<\mathcal{T} \\
    F_{\mathcal{T}} - F_{T_{max}}+F_{T_{max}-\mathcal{T}+2} =F_2 & \text{if } T_{max}=\mathcal{T}
    \end{cases}
\end{align*}
Besides, since $F_1<F_{\mathcal{T}}$, then we cannot have $T_{max}=1$ and we have
\begin{align*}
    V(F)_{\mathcal{T}-1}
    =F_{\mathcal{T}}-F_{T_{max}}+F_{T_{max}-\mathcal{T}+\mathcal{T}-1} =F_{\mathcal{T}}+(F_{T_{max}-1}-F_{T_{max}})
    \leq F_{\mathcal{T}}=V(F)_{\mathcal{T}}
\end{align*}
Hence, assuming that $T_{max}<\mathcal{T}$, Point (ii) is satisfied.

Finally, we have for $1\le t\le \mathcal{T}-1$
\begin{equation*}
    V(F)_{\mathcal{T}}-V(F)_t
    =\begin{cases}
        (F_{\mathcal{T}}-F_1)+(F_{T_{max}}-F_{T_{max}+t-1}) \ge 0, &\mbox{if }1 \leq t \leq 1+\mathcal{T}-T_{max},\\
   (F_{T_{max}}-F_{T_{max}-\mathcal{T}+t}) \ge 0, &\mbox{if }  1+\mathcal{T}-T_{max}<t\leq \mathcal{T}-1,
    \end{cases}
\end{equation*}
which proves Point (iii).
\end{proof}

\section{Analysis of the rejection step}\label{appen:rejection}

Algorithm~\ref{alg1} relies on a rejection step to generate trajectories between two records. An approximation of the rejection rate can be obtained by recalling that the discrete bridge $\lbrace Y_t\rbrace_{t\in\interv{1}{\mathcal{T}}}$ is built from a discretization of a Brownian bridge $\lbrace \tilde{B}_t\rbrace_{t\in [0,T]}$, $T=\cT-1$, through~\eqref{eq:bridge_mod}. Hence, to evaluate the probability that $T_{max}\le \cT-1$, we can evaluate the probability that the maximum of the Brownian bridge $\lbrace \tilde{B}_t\rbrace_{t\in [0,T]}$ occurs before the time $\cT-2=T-1$. This probability can be evaluated using the next proposition.

\begin{proposition}\label{prop:prob_rej}
    Let $T>0$ and $\lbrace \tilde{B}_t\rbrace_{t\in [0,T]}$ be a Brownian bridge between the value $\tilde B_0=R_1$ and $\tilde B_{T}=R_2>0$, with variance $\sigma>0$. Let $M=\max_{t\in [0,T]}\tilde{B}_t$ and $\theta=\sup\lbrace t\in[0,T] : B_t= M\rbrace$. Then, for any $s \in (0,1)$
    \begin{equation}\label{eq:prob_rej}
        \mathbb{P}(\theta<sT) = 2 s\left(\Phi(-\mu_{s})(1+\mu_s^2)-\mu_s\phi(\mu_s)\right), \quad \text{with } \mu_s = \frac{R_2-R_1}{\sigma}\frac{1-s}{\sqrt{Ts(1-s)}},
    \end{equation}
    where $\phi(\cdot)$ (resp. $\Phi(\cdot)$) denotes the density (resp. cumulative distribution function) of a standard Gaussian distribution. 
\end{proposition}

\begin{proof}
    Without loss of generality, let us assume that $R_1=0$, and that $R_2=a>0$. The general case $R_1\neq 0$ can then be obtained by simple translation.

    A Brownian bridge between the value $\tilde B_0=0$ and $\tilde B_{T}=a>0$, with variance $\sigma^2>0$ is a Brownian motion with variance $\sigma^2$, conditioned to take the value $\tilde B_{T}=a>0$. Let then $\lbrace W_t\rbrace_{t\in [0,T]}$ with unit-variance, so that $\lbrace \sigma W_t\rbrace_{t\in [0,T]}$ defines a Brownian motion with variance $\sigma^2$, and take $m=\max_{t\in [0,T]}\sigma W_t=\sigma \max_{t\in [0,T]} W_t$ and $\eta=\sup\lbrace t\in[0,T] : \sigma W_t= m\rbrace=\sup\lbrace t\in[0,T] :  W_t= m/\sigma\rbrace$. Hence, we have
    \begin{equation*}
        \mathbb{P}(\theta<sT) = \mathbb{P}(\eta<sT\vert  W_{T} = \hat a)
    \end{equation*}
    and more generally,  for $b\ge a$, and $\dd b$ an infinitesimal interval around $b$, 
    \begin{align*}
        \mathbb{P}(\theta<sT, M\in \dd b)=\mathbb{P}(\eta<sT,m \in \dd b \vert \sigma W_{T} = a)
        =\mathbb{P}(\eta<sT, \hat m\in \dd\hat b \vert  W_{T} = \hat a)
    \end{align*}
    where $\hat m = m/\sigma = \max_{t\in [0,T]} W_t$, $\hat b= b/\sigma$ and $\hat a = a/\sigma$. This last probability can be computed by noting that, following \cite[p .101]{karatzas2014brownian}
     \begin{align*}
        \mathbb{P}(\eta<sT, \hat m\in \dd\hat b ,  W_{T} \in\dd \hat a)
        =\frac{2}{\sqrt{T^3}}\left[ 
        (2\hat{b}-\hat{a})\Phi\left(-\frac{\mu_+}{\chi}\right)\phi\left(\frac{2\hat{b}-\hat{a}}{\sqrt{T}}\right)
        -\hat{a}\Phi\left(-\frac{\mu_-}{\chi}\right)\phi\left(\frac{\hat{a}}{\sqrt{T}}\right)
        \right]\dd\hat{a}\dd\hat{b}
    \end{align*}
    where
    \begin{equation*}
        \mu_{\pm}=\frac{\hat b(T-S)\pm (\hat a-\hat b)S}{T}, \quad S=sT, \quad \chi^2=\frac{S(T-S)}{T}
    \end{equation*}
    Hence,
         \begin{align*}
        \mathbb{P}(\eta<sT,  W_{T} \in\dd \hat a)
        &=\frac{2\dd\hat{a}}{\sqrt{T^3}}\int_a^\infty\left[ 
        (2\hat{b}-\hat{a})\Phi\left(-\frac{\mu_+}{\chi}\right)\phi\left(\frac{2\hat{b}-\hat{a}}{\sqrt{T}}\right)
        -\hat{a}\Phi\left(-\frac{\mu_-}{\chi}\right)\phi\left(\frac{\hat{a}}{\sqrt{T}}\right)
        \right]\dd\hat{b} \\
        &=\frac{2\dd\hat{a}}{\sqrt{T^3}}\left[ 
        \int_a^\infty(2\hat{b}-\hat{a})\Phi\left(-\frac{\hat{b}(T-2S)+\hat as}{T\chi}\right)\phi\left(\frac{2\hat{b}-\hat{a}}{\sqrt{T}}\right)\dd\hat{b}
        -\hat{a}\phi\left(\frac{\hat{a}}{\sqrt{T}}\right)\int_a^\infty \Phi\left(-\frac{\hat bt - \hat as}{T\chi}\right)\dd\hat{b} 
        \right]\\
        &=\frac{2\dd\hat{a}}{\sqrt{T^3}}\left[ 
        \frac{T}{2}\int_{a/\sqrt{T}}^\infty u~\Phi\left(-\frac{u(T-2S)+\hat a \sqrt{T}}{2\sqrt{T}\chi}\right)\phi\left(u\right)\dd u
        -\hat{a}\phi\left(\frac{\hat{a}}{\sqrt{T}}\right)\int_a^\infty \Phi\left(-\frac{\hat bt - \hat as}{T\chi}\right)\dd\hat{b} 
        \right]\\
        &=\frac{2\dd\hat{a}}{\sqrt{T^3}}\bigg[
        \frac{T}{2}\left(
        -\frac{T-2S}{T}\phi\left(\frac{\hat a}{\sqrt{T}}\right)\left(1-\Phi\left(\frac{\hat{a}(T-S)}{T\chi}\right)\right)+\phi\left(\frac{\hat a}{\sqrt{T}}\right)\Phi\left(-\frac{\hat{a}(T-S)}{T\chi}\right)\right)\\
        & \quad\quad -\hat{a}\phi\left(\frac{\hat{a}}{\sqrt{T}}\right)
        \chi\left(
        \left(-\frac{\hat{a}(T-S)}{T\chi}\right)\Phi\left(-\frac{\hat{a}(T-S)}{T\chi}\right)+\phi\left(-\frac{\hat{a}(T-S)}{T\chi}\right)
        \right)
        \bigg]
    \end{align*}
where we used the formulas in~\cite[Formulae 10,011.2 and 10,000]{owen1980table} to compute the two integrals. Taking then
\begin{equation*}
    \mu_s = \frac{\hat{a}(T-S)}{T\chi}=\hat{a}\left(1-\frac{S}{T}\right)\frac{\sqrt{T}}{\sqrt{S(T-S)}}=\hat{a}\frac{1-s}{\sqrt{Ts(1-s)}}=\hat{a}\sqrt{\frac{1-s}{sT}}
\end{equation*}
we obtain
    \begin{align*}
        \mathbb{P}(\eta<sT,  W_{T} \in\dd \hat a)
        &=\frac{2\dd\hat{a}}{\sqrt{T^3}}\phi\left(\frac{\hat a}{\sqrt{T}}\right)\bigg[
        -\frac{T-2S}{2}\left(1-\Phi\left(\mu_s\right)\right)+\frac{T}{2}\Phi\left(-\mu_s\right)
          -\hat{a}
        \chi\left(
        -\mu_s\Phi\left(-\mu_s\right)+\phi\left(\mu_s\right)
        \right)
        \bigg]
    \end{align*}
Then, noting that $\chi\hat{a}=\sqrt{Ts(1-s)}\hat{a}=Ts \mu_s=S\mu_s$
and that $\Phi\left(\mu_s\right)=1-\Phi\left(-\mu_s\right)$ (by symmetry of the $\phi$), we get
    \begin{align*}
        \mathbb{P}(\eta<sT,  W_{T} \in\dd \hat a)
        &=\frac{2\dd\hat{a}}{\sqrt{T^3}}\phi\left(\frac{\hat a}{\sqrt{T}}\right)\bigg[
        S(1+\mu_s^2)\Phi\left(-\mu_s\right)
          -S\phi\left(\mu_s\right)
        \bigg]\\
        &=2s\bigg[
       (1+\mu_s^2)\Phi\left(-\mu_s\right)
          -\phi\left(\mu_s\right)\bigg]
          \frac{1}{\sqrt{T}}\phi\left(\frac{\hat a}{\sqrt{T}}\right)
        \dd\hat{a}
    \end{align*}
Finally,  since $W_T$ is normally distributed with mean $0$ and variance $T$, we have
\begin{equation*}
    \mathbb{P}( W_{T} \in\dd \hat a) = \frac{1}{\sqrt{2\pi T}}\exp\left(-\frac{\hat{a}^2}{2T}\right)\dd\hat{a}=\frac{1}{\sqrt{T}}\phi\left(\frac{\hat a}{\sqrt{T}}\right)
        \dd\hat{a}
\end{equation*}
which allows us to conclude that
    \begin{align*}
        \mathbb{P}(\theta<sT) =\mathbb{P}(\eta<sT \vert W_{T} = \hat a)
        &=\frac{\mathbb{P}(\eta<sT,  W_{T} \in\dd \hat a)}{\mathbb{P}( W_{T} \in\dd \hat a)}
        =2s\bigg[
       (1+\mu_s^2)\Phi\left(-\mu_s\right)
          -\phi\left(\mu_s\right)\bigg]
    \end{align*}
which concludes the proof.

\end{proof}

Circling back to the rejection step, the probability that $T_{max}<\cT$ can be approached by its continuous counterpart $\mathbb{P}(\theta<sT)$ in Proposition~\eqref{prop:prob_rej}, with $s=(T-1)/T=(\cT-2)/(\cT-1)$. Note in particular that this probability can be made larger by making $\mu_s$ approaching zero, which can be achieved by either working with a large(r) time series size $\cT=T+1$, or by working with larger variances $\sigma^2$ for the Brownian bridge.

\section{Properties of the smoothing}\label{appen:smooth}

The goal of this section is to show that the smoothing steps yielding the time series  $\lbrace X_t\rbrace_{t\in\interv{1}{\mathcal{T}}}$ in \eqref{eq:w_s} do not affect the two key properties of the orginal time series  $\lbrace Z_t\rbrace_{t\in\interv{1}{\mathcal{T}}}$, namely that its endpoints are $X_1=Z_1=R_1$ and $X_{\cT}=Z_{\cT}=R_2$, and are also records. 

The parameter $\alpha_{\mathcal{T}}$ introduced in~\eqref{eq:def_hatz} plays a key role in this regard, as it ensures that the intermediate smoothing $\lbrace\widetilde{Z}_t\rbrace_{t\in\interv{1} {\mathcal{T}}}$ stays sufficiently below the value of the first record $Z_1=R_1$ near $t=1$ to ensure that the final smoothing $\lbrace X_t\rbrace_{t\in\interv{1}{\mathcal{T}}}$ can admit a record at $t=1$. We note that $\lbrace\widetilde{Z}_t\rbrace_{t\in\interv{1} {\mathcal{T}}}$ is modified only when
    \begin{equation*}
        \widetilde{Z}_2 > \frac{\alpha_{\mathcal{T}}}{1-w^{(s)}(\mathcal{T}-1)}Z_1
    \end{equation*}
    where by definition of $\alpha_{\cT}$,
    \begin{equation}\label{eq:aw}
        \frac{\alpha_{\mathcal{T}}}{1-w^{(s)}(\mathcal{T}-1)}=\frac{1-w^{(s)}(\mathcal{T}-3)}{1-w^{(s)}(\mathcal{T}-1)}+\frac{w^{(s)}(\mathcal{T}-3)}{1-w^{(s)}(\mathcal{T}-1)}\exp\left(-\frac{2\cT-3}{2s^2}\right).
    \end{equation}
    In practice, for small smoothing parameters $(s/\cT)\ll 1$, the quantity~\eqref{eq:aw} becomes close to $1$, while for large smoothing parameters $(s/\cT)\gg 1$, the second term in~\eqref{eq:aw} explodes. 

Then, as detailed in the proposition below, a simple condition can be checked to ensure that  $(R_1,1)$ and $(R_2,\mathcal{T})$ are also records $\lbrace X_t\rbrace_{t\in\interv{1}{\mathcal{T}}}$: namely, the ratio between the record values should satisfy the constraint:
    \begin{equation}\label{eq:cond_rec}
        \frac{R_2}{R_1} < M_{\cT}^{(s)}:=\frac{1}{1-w^{(s)}(1)}\left(1-\frac{w^{(s)}(\mathcal{T}-1)}{w^{(s)}(\mathcal{T}-2)}\right) + \frac{1-\alpha_{\mathcal{T}}}{w^{(s)}(\mathcal{T}-2)}(1-w^{(s)}(\mathcal{T}-2)).
    \end{equation}
This requirement may seem restrictive at first sight, but is generally satisfied in practice due to two mechanisms. Indeed, note that the bound $M_{\cT}^{(s)}$ in~\eqref{eq:cond_rec} can be rewritten as
\begin{align*}
M_{\cT}^{(s)} &= \frac{1}{1-w^{(s)}(1)}\left(1-\frac{w^{(s)}(\mathcal{T}-1)}{w^{(s)}(\mathcal{T}-2)}\right) + \frac{w^{(s)}(\mathcal{T}-3)}{w^{(s)}(\mathcal{T}-2)}\left(1-\frac{w^{(s)}(\mathcal{T}-1)}{w^{(s)}(\mathcal{T}-2)}\right)(1-w^{(s)}(\mathcal{T}-2))\\
&= \left(1-\exp\left(-\frac{2\cT-3}{2s^2}\right)\right) \left[\frac{1}{1-w^{(s)}(1)}
 + \exp\left(\frac{2\cT-5}{2s^2}\right)(1-w^{(s)}(\mathcal{T}-2))\right].
\end{align*}
Hence, for small smoothing parameters $(s/\cT)\ll 1$, we have $(2\mathcal{T}-5)/(2s^2) \gg 1$, $w^{(s)}(\cT-2)\approx 0$, and therefore  $M_{\cT}^{(s)}\rightarrow \infty$, meaning that asymptotically, the constraint~\eqref{eq:cond_rec} is always satisfied.
Conversely, for large smoothing parameters $(s/\cT)\gg 1$, we have $w^{(s)}(1)\approx 1 - 1/(2s^2)$ and therefore $M_{\cT}^{(s)}\approx 2\cT-3$. Hence, working with a larger value $\mathcal{T}$ (i.e. with a more discretized time series) allows to make the constraint more easily satisfied.

\begin{proposition}
    Let   $\{Z_t\}_{t\in\interv{1}{\mathcal{T}}}$ be a time series admitting as endpoints and records the pairs $(R_1,1)$ and $(R_2,\mathcal{T})$ (with $R_2>R_1$). 
    If the ratio $R_2/R_1$ satisfies~\eqref{eq:cond_rec},
    then the smoothed time series $\{X_t\}_{t\in\interv{1}{\mathcal{T}}}$ defined in~\eqref{eq:w_s} also admits as endpoints and records $X_1=R_1$ and $X_{\mathcal{T}}=R_2$.
\end{proposition}

\begin{proof}
    It is possible to check that, by definition~\eqref{eq:w_s}, we have $X_1=Z_1=R_1$ and $X_{\mathcal{T}}=Z_{\mathcal{T}}=R_2$. Let us now show that these endpoints are records.

    On the one hand, for $t=1$, since $X_1=Z_1=R_1$ and $X_\mathcal{T}=Z_\mathcal{T}=R_2$, we have
    \begin{align*}
        X_2-X_1
        &=w^{(s)}(1)\frac{1-w^{(s)}(\mathcal{T}-2)}{1-w^{(s)}(\mathcal{T}-1)}Z_1 + \left(1-w^{(s)}(\mathcal{T}-2)\right)\left(1-w^{(s)}(1)\right)\widehat{Z}_2+w^{(s)}(\mathcal{T}-2)\frac{1-w^{(s)}(1)}{1-w^{(s)}(\mathcal{T}-1)}Z_{\mathcal{T}} -Z_1 \\
        &\le w^{(s)}(1)\frac{1-w^{(s)}(\mathcal{T}-2)}{1-w^{(s)}(\mathcal{T}-1)}R_1 + \left(1-w^{(s)}(\mathcal{T}-2)\right)\left(1-w^{(s)}(1)\right)\frac{\alpha_{\mathcal{T}}}{1-w^{(s)}(\mathcal{T}-1)}R_1+w^{(s)}(\mathcal{T}-2)\frac{1-w^{(s)}(1)}{1-w^{(s)}(\mathcal{T}-1)}R_2-R_1\\
        &=\frac{R_1}{1-w^{(s)}(\mathcal{T}-1)}\left((1-w^{(s)}(\mathcal{T}-2))\left(w^{(s)}(1)+\alpha_{\mathcal{T}}(1-w^{(s)}(1))\right)+w^{(s)}(\mathcal{T}-2)(1-w^{(s)}(1))\frac{R_2}{R_1}-(1-w^{(s)}(\mathcal{T}-1))\right)\\
        &=\frac{R_1
        w^{(s)}(\mathcal{T}-2)(1-w^{(s)}(1))}{1-w^{(s)}(\mathcal{T}-1)}\left(\frac{R_2}{R_1} +
        \frac{1-w^{(s)}(\mathcal{T}-2)}{
        w^{(s)}(\mathcal{T}-2)(1-w^{(s)}(1))}\left(1-(1-\alpha_{\mathcal{T}})(1-w^{(s)}(1))\right)-\frac{1-w^{(s)}(\mathcal{T}-1)}{
        w^{(s)}(\mathcal{T}-2)(1-w^{(s)}(1))}\right)\\
        &=\frac{R_1
        w^{(s)}(\mathcal{T}-2)(1-w^{(s)}(1))}{1-w^{(s)}(\mathcal{T}-1)}\left(\frac{R_2}{R_1} -
        \frac{(1-w^{(s)}(\mathcal{T}-2))(1-\alpha_{\mathcal{T}})}{
        w^{(s)}(\mathcal{T}-2)}-\frac{w^{(s)}(\mathcal{T}-2)-w^{(s)}(\mathcal{T}-1)}{
        w^{(s)}(\mathcal{T}-2)(1-w^{(s)}(1))}\right)<0\\
    \end{align*}
    where we used \eqref{eq:cond_rec} to derive the last inequality. Consequently, we have $X_2<X_1$, which confirms that $X_1$ is a record of $\{X_t\}_{t\in\interv{1}{\mathcal{T}}}$.

    Similarly, for $t=\mathcal{T}$,
    \begin{equation*}
   X_{\mathcal{T}-1}
   =w^{(s)}(\mathcal{T}-2)\frac{1-w^{(s)}(1)}{1-w^{(s)}(\mathcal{T}-1)}R_1 + \left(1-w^{(s)}(1)\right)\left(1-w^{(s)}(\mathcal{T}-2)\right)\widetilde{Z}_t+w^{(s)}(1)\frac{1-w^{(s)}(\mathcal{T}-2)}{1-w^{(s)}(\mathcal{T}-1)}Z_{\mathcal{T}} 
\end{equation*}
    In particular since $Z_{\mathcal{T}}$ is a record of  $\{Z_t\}_{t\in\interv{1}{\mathcal{T}}}$, we have for any $t\le \mathcal{T}$, $Z_t\le Z_{\mathcal{T}}$ and therefore $\widetilde{Z}_t\le Z_{\mathcal{T}}=R_2$ (since it is defined as a convex combination of $\{Z_t\}_{t\in\interv{1}{\mathcal{T}}}$). Besides, we also have by assumption $R_1 < R_2$. Hence,
    \begin{align*}
   X_{\mathcal{T}-1}
   &<w^{(s)}(\mathcal{T}-2)\frac{1-w^{(s)}(1)}{1-w^{(s)}(\mathcal{T}-1)}R_2 + \left(1-w^{(s)}(1)\right)\left(1-w^{(s)}(\mathcal{T}-2)\right)R_2+w^{(s)}(1)\frac{1-w^{(s)}(\mathcal{T}-2)}{1-w^{(s)}(\mathcal{T}-1)}R_2\\
   &=\frac{R_2}{1-w^{(s)}(\mathcal{T}-1)}\left(w^{(s)}(\mathcal{T}-2)(1-w^{(s)}(1)) + \left(1-w^{(s)}(1)\right)\left(1-w^{(s)}(\mathcal{T}-2)\right)(1-w^{(s)}(\mathcal{T}-1))+w^{(s)}(1)(1-w^{(s)}(\mathcal{T}-2))\right)
   \\
   &=\frac{R_2}{1-w^{(s)}(\mathcal{T}-1)}\left(1-w^{(s)}(1)w^{(s)}(\mathcal{T}-2)-\left(1-w^{(s)}(1)\right)\left(1-w^{(s)}(\mathcal{T}-2)\right)w^{(s)}(\mathcal{T}-1)\right)
   \\
   &=\frac{R_2}{1-w^{(s)}(\mathcal{T}-1)}\left(1-w^{(s)}(\mathcal{T}-1)-w^{(s)}(1)w^{(s)}(\mathcal{T}-2)+w^{(s)}(\mathcal{T}-1)\left(1-\left(1-w^{(s)}(1)\right)\left(1-w^{(s)}(\mathcal{T}-2)\right)\right)\right)
      \\
   &=\frac{R_2}{1-w^{(s)}(\mathcal{T}-1)}\left(1-w^{(s)}(\mathcal{T}-1)-w^{(s)}(1)w^{(s)}(\mathcal{T}-2)+w^{(s)}(\mathcal{T}-1)\left(1-\left(1-w^{(s)}(1)\right)\left(1-w^{(s)}(\mathcal{T}-2)\right)\right)\right)
         \\
   &=R_2\left(1-\frac{w^{(s)}(1)w^{(s)}(\mathcal{T}-2)+w^{(s)}(\mathcal{T}-1)\left(w^{(s)}(1)w^{(s)}(\mathcal{T}-2)-w^{(s)}(1)-w^{(s)}(\mathcal{T}-2)\right)}{1-w^{(s)}(\mathcal{T}-1)}\right)=R_2\left(1-\frac{\beta}{1-w^{(s)}(\mathcal{T}-1)}\right)
\end{align*}
where
\begin{align*}
    \beta
    &=w^{(s)}(1)w^{(s)}(\mathcal{T}-2)+w^{(s)}(\mathcal{T}-1)\left(w^{(s)}(1)w^{(s)}(\mathcal{T}-2)-w^{(s)}(1)-w^{(s)}(\mathcal{T}-2)\right)\\
    &= w^{(s)}(1)w^{(s)}(\mathcal{T}-2)w^{(s)}(\mathcal{T}-1)\left[\frac{1}{w^{(s)}(\mathcal{T}-1)}+1-\frac{1}{w^{(s)}(1)}-\frac{1}{w^{(s)}(\mathcal{T}-2)})\right]\\
    &= w^{(s)}(1)w^{(s)}(\mathcal{T}-2)w^{(s)}(\mathcal{T}-1)\left[{v^{(s)}(\mathcal{T}-1)}-{v^{(s)}(\mathcal{T}-2)}-\left(v^{(s)}(1)-{v^{(s)}(0)}\right)\right]
\end{align*}
where we define for $k\ge 0$, $v^{(s)}(k)=1/w^{(s)}(k)=\exp(k^2/2s^2)$. Note in particular that $k\mapsto v^{(s)}(k)$ is a convex function and therefore satisfies, for any $0\le k_1<k_2<k_3$ the inequality
\begin{equation*}
\frac{v^{(s)}(k_2)-v^{(s)}(k_1)}{k_2-k_1}
\le \frac{v^{(s)}(k_3)-v^{(s)}(k_1)}{k_3-k_1}
\le \frac{v^{(s)}(k_3)-v^{(s)}(k_2)}{k_3-k_2}
\end{equation*}
Hence, applying  these inequalities to  $0<1<\mathcal{T}-2$ and $1<\mathcal{T}-2<\mathcal{T}-1$ we obtain
\begin{equation*}
\frac{v^{(s)}(1)-v^{(s)}(0)}{1-0}
\le \frac{v^{(s)}(\mathcal{T}-2)-v^{(s)}(1)}{\mathcal{T}-2-1}
\le \frac{v^{(s)}(\mathcal{T}-1)-v^{(s)}(\mathcal{T}-2)}{\mathcal{T}-1-(\mathcal{T}-2)}
\end{equation*}
which gives $v^{(s)}(1)-v^{(s)}(0)\le v^{(s)}(\mathcal{T}-1)-v^{(s)}(\mathcal{T}-2)$. This in turn allows us to conclude that $\beta\ge 0$ and therefore that $X_{\mathcal{T}-1}<R_2$ : $X_{\mathcal{T}}$ is therefore a record.
\end{proof}

\section{Properties of the times series generated by Algorithm 2}\label{appen:res_alg2}

\begin{proposition}
    Let  $\{X_t\}_{t\in\interv{1}{\mathcal{T}}}$ be a random trajectory generated using Algorithm~\ref{alg2} with a record set $\mathcal{M}=\lbrace (R_n,L_n)\rbrace_{n=1,\dots,N}$ (with $L_1=1$, $L_N=\mathcal{T}$), a standard-deviation of Gaussian increments $\sigma>0$, and a smoothing parameter $s\ge 0$. 

    Then, the endpoints of the time series are $X_1=R_1$ and $X_{\mathcal{T}}=R_N$, and we have $\mathcal{M}\subset \mathcal{R}(\{X_t\}_{t\in\interv{1}{\mathcal{T}}})$.
\end{proposition}

\begin{proof}

Let us first check that the concatenation of the time series $\{S^n_1,...,S^n_{l_n-1}\}$, for $n \in \{1,...,N-2\}$ and of $\{S^{N-1}_1,...,S^{N-1}_{l_N}\}$ made in Algorithm~\ref{alg2} results indeed in a time series of length $\mathcal{T}$. Hence, we denote by $H$ the length of the concatenation. Then,
\begin{align*}
    H&=\biggl(\sum_{n=1}^{N-1}length(X^n)-1\biggl)+1\\
    &=\biggl(\sum_{n=1}^{N-1}l_n\biggl)-(N-1)+1\\
    &=\biggl(\sum_{n=1}^{N-1}L_{n+1}-L_n+1\biggl)-N+2\\
    &=\biggl(\sum_{n=1}^{N-1}L_{n+1} -\sum_{n=1}^{N-1}L_{n}+\sum_{n=1}^{N-1}1\biggl)-N+2\\
    &=\biggl(\sum_{n=2}^{N}L_{n}\biggl)-\biggl(\sum_{n=1}^{N-1}L_{n}\biggl)+N-1-N+2\\
    &=L_N-L_1+1\\
    &=\mathcal{T}
    \end{align*}
by definition of $L_1$ and $L_N$. Hence, as constructed by Algorithm~\ref{alg2}, $(X_t)_t$ is indeed a time series of length $\mathcal{T}$. Besides, by definition of the first and last time series in the concatenation, we obtain directly that $X_1=R_1$ and $X_{\mathcal{T}}=R_N$, and following the properties of the Vervaat transform (cf. Proposition~\ref{prop:vervaat}) and the definition of the smoothing~\eqref{eq:w_s}, we have $X_1 \ge X_2$ and $X_{\mathcal{T}}\ge X_{\mathcal{T}-1}$ which means that the endpoints of $\{X_t\}_{t\in\interv{1}{\mathcal{T}}}$ are local maxima, and therefore records of the time series.

We now prove that $\mathcal{M}\subset \mathcal{R}(\{X_t\}_{t\in\interv{1}{\mathcal{T}}})$. Given the result above, we only need to check that for $1<n<N$, the record pair $(R_n,L_n)\in\mathcal{M}$, satisfies $X_{L_n}=R_n$ and $ X_{L_n}\geq \max(X_1,...,X_{L_n-1},X_{L_n+1}) $. Once again, following the properties of the Vervaat transform (cf. Proposition~\ref{prop:vervaat}) and the definition of the smoothing~\eqref{eq:w_s}, we have that $X_{L_n}=R_n\ge X_{L_n+1}$, $X_{L_n+l_n}=X_{L_{n+1}}=R_{n+1} \ge X_{L_{n+1}-1}$, $ X_{L_{n+1}}\geq \max(X_{L_n},...,X_{L_{n+1}-1})$. We can therefore conclude by induction that, 
$$ X_{L_{n+1}}=R_{n+1}\geq  \max(X_{L_n},...,X_{L_{n+1}-1})\ge X_{L_n} =R_n \ge \max(X_{L_{n-1}},...,X_{L_{n}-1}) \ge \dots \ge \max(X_{L_{1}},...,X_{L_{2}-1})$$
Hence, for any $1<n<N$, we have $X_{L_{n}}=R_{n}\geq  \max(X_{L_1},...,X_{L_{n}-1})=\max(X_{1},...,X_{L_{n}-1})$ and $X_{L_n} \ge X_{L_n+1}$, meaning that $(R_n,L_n)$ is indeed a record of $\{X_t\}_{t\in\interv{1}{\mathcal{T}}}$. This concludes the proof.

\end{proof}

\section{Metrics used}\label{appen:metrics}
$Z=(z_1,z_2,...z_\mathcal{T})$ a simulated trajectory and $X=(x_1,x_2,...,x_\mathcal{T})$ an observed one with possible missing values. We denote by $I$ the indexes where values of $X$ are not missing.\\

i) MSE
\begin{equation}\label{mse}
    \frac{1}{\mathcal{T}}\sum_{i\in I}(z_i-x_i)^2
\end{equation}

ii)MAPE
\begin{equation}\label{mape}
    \frac{1}{\mathcal{T}}\sum_{i\in I}\left\lvert\frac{z_i-x_i}{x_i}\right\lvert
\end{equation}

iii) Variance
\begin{equation}\label{var}
    \frac{1}{\mathcal{T}}\sum_{i=1}^{\mathcal{T}}(z_i-\bar{z})^2
\end{equation}

iv) Area under the curve
\begin{equation}\label{auc}
    \sum_{i=1}^{\mathcal{T}-1}\frac{\Delta t}{2}+(z_{i+1}+{z_i})
\end{equation}

\begin{center}
    \includegraphics[width=0.7\linewidth]{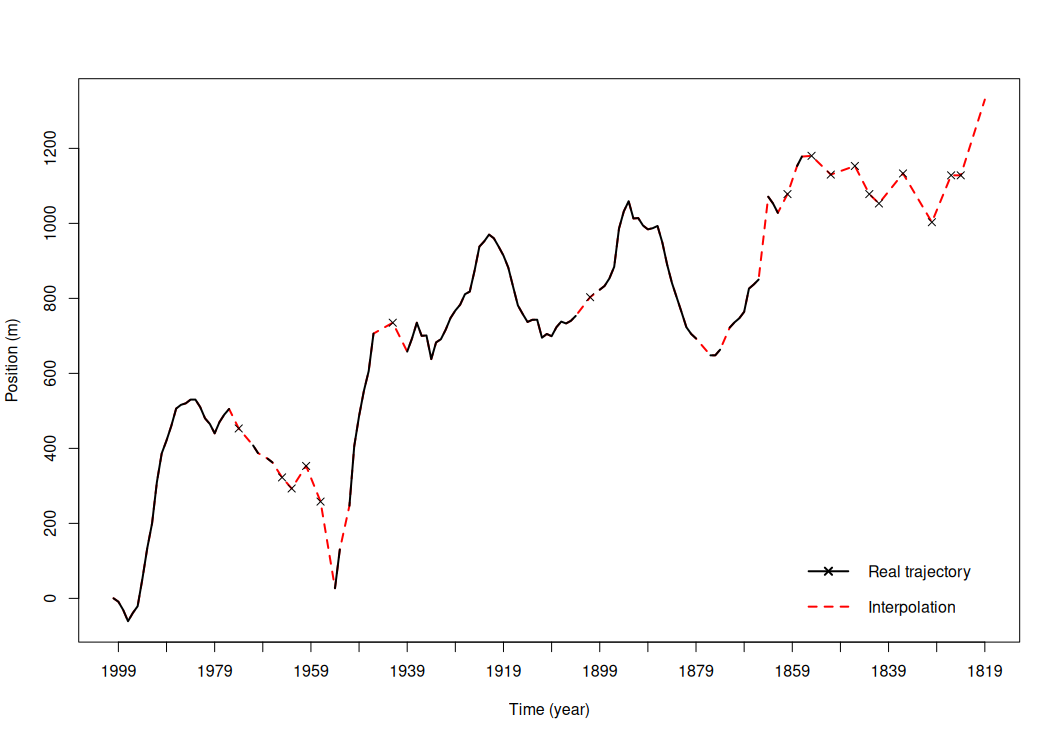}
    \captionof{figure}{Result of the linear interpolation used to compute metrics of Figure\ref{fig:4metrics}. Computed with the \textit{na.approx} function of the R package \textit{zoo}. }
    \label{fig:interpolation}
\end{center}

\bibliography{references2}

\end{document}